\documentclass[11pt,twoside]{article}
\usepackage{mathrsfs}
\usepackage{amsmath,amssymb}
\usepackage{amsfonts}
\usepackage{latexsym,bm}
\usepackage{amsthm}
\usepackage{amssymb,amscd}
\usepackage{amsbsy}
\usepackage{fancyhdr,graphicx}
\usepackage[dvips]{psfrag}
\usepackage{indentfirst}
\usepackage{graphics,color}
\usepackage{epsfig}
\usepackage{subfigure}
\usepackage{xcolor}

\numberwithin{equation}{section}

\newtheorem {theorem} {Theorem}[section]
\newtheorem {proposition} [theorem]{Proposition}

\theoremstyle{remark}

\theoremstyle{definition}

\newtheorem {remark} [theorem]{Remark}

\begin{document}
\setlength{\parindent}{4ex}
\setlength{\parskip}{1ex}
\setlength{\oddsidemargin}{12mm}
\setlength{\evensidemargin}{9mm}

\title
{  Bifurcation analysis and global dynamics  of a mathematical model of  antibiotic resistance in hospitals}
\begin{figure}[b]
\rule[-0.5ex]{7cm}{0.2pt}\\
\\E-mail address: cenxiuli2010@163.com (X. Cen), zfeng@math.purdue.edu (Z. Feng), zheng30@math.purdue.edu (Y. Zheng), mcszyl@mail.sysu.edu.cn (Y. Zhao).\\
Supported by the NSF of China (No.11171355 and No.11401111) and the State Scholarship Fund of CSC (No. 201208440200)
\end{figure}
\author
{ {  Xiuli Cen$^{a}$, Zhilan Feng$^{b}$, Yiqiang Zheng$^{b}$ and Yulin Zhao$^{c}$}\\
{\footnotesize\it $^{a}$Department of Mathematical Sciences, Tsinghua University, Beijing, 100084, P.R.China}\\
{\footnotesize\it $^{b}$Department of Mathematics, Purdue University, West Lafayette, IN 47907, USA}\\
{\footnotesize\it $^{c}$Department of Mathematics, Sun Yat-sen University, Guangzhou, 510275, P.R.China}}

\date{}
\maketitle {\narrower \small \noindent {\bf Abstract\,\,\,}
Antibiotic-resistant bacteria has posed a grave threat to public health by causing a number of nosocomial infections in hospitals. Mathematical models have been used to study the transmission of antibiotic-resistant bacteria within a hospital and the measures to control antibiotic resistance in nosocomial pathogens. Studies presented in \cite{LBL,LB} have shown great value in understanding the transmission of antibiotic-resistant bacteria in a hospital. However, their results are limited to numerical simulations of a few different scenarios without analytical analysis of the models in all biologically feasible parameter regions. Bifurcation analysis and identification of the global stability conditions are necessary to assess the interventions which are proposed to limit nosocomial infection and stem the spread of antibiotic-resistant bacteria. In this paper we study the global dynamics of the mathematical model of antibiotic resistance in hospitals in \cite{LBL,LB}. The invasion reproduction number $\mathcal R_{ar}$ of antibiotic-resistant bacteria is introduced. We give the relationship of $\mathcal R_{ar}$ and two control reproduction numbers of sensitive bacteria and resistant bacteria ($\mathcal R_{sc}$ and $\mathcal R_{rc}$). More importantly, we prove that a backward bifurcation may occur at $\mathcal R_{ar}=1$ when the model includes superinfection which is not mentioned in \cite{LB}. That is, there exists a new threshold $\mathcal R_{ar}^c$, and if $\mathcal R_{ar}^c<\mathcal R_{ar}<1$, then the system can have two interior equilibria and it supports an interesting bistable phenomenon. This provides critical information on controlling the antibiotic-resistance in a hospital. }

\vskip 0.2cm

{\it Keywords}:  Antibiotic resistance; Invasion reproduction number; Backward bifurcation; Bistable phenomenon; Global dynamics.

\section{Introduction}

Antimicrobial resistance brings a huge threat to
the effective prevention and treatment of an ever-increasing range
of infections caused by antibiotic-resistant bacteria. In general, patients with infections are
at higher risk of worse clinical outcomes and even death,
and they also consume more healthcare resources. It is one of the preeminent public health
concerns in the 21st century. Mathematical models have made substantial contributions to explain
antibiotic resistance in hospitals as they provide quantitative criteria to
evaluate the interventions to control nosocomial
infection and stem the spread of antibiotic-resistant bacteria.

There has been a substantial amount of work devoted to understanding
the dynamics of antibiotic resistance in hospitals, see for instance
\cite{BBD,BLL,C,LBL,LB,LHL,MO,WAMR}. Lipsitch et al. \cite{LBL} proposed a mathematical model of the transmission
dynamics of antibiotic-resistant and sensitive strains of a type of bacteria in a hospital
or a unit of hospital. They studied the transmission dynamics of
these two strains and the use of two antimicrobial agents referred to
drug 1 and drug 2, and assumed that the bacteria may be sensitive or resistance to drug 1, but all bacteria were sensitive to drug 2. The authors
predicted the transmission dynamics of resistant
and sensitive nosocomial pathogens, and suggested criteria
for measuring the effectiveness of interventions to reduce resistance in
hospitals based on an ODEs model. With numerical simulations, they predicted
how the prevalence of resistant bacteria changes over time under
various interventions. Furthermore, the same authors extended the model in \cite{LB} by allowing
superinfection, which is ignored in \cite{LBL}. The results in \cite{LBL,LB} are limited to
numerical analysis. However, analytical studies of the models are necessary to provide a full assessment of
possible interventions.

We provided analytical analysis of the model \cite{LB} in this paper.
The model considers two strains of a single bacterial species with
two antimicrobial agents referred to drug 1 and drug 2. Individuals
may carry strains of these bacteria that are either sensitive ($S$)
or resistant ($R$) to drug 1, or they may be free of these bacteria
($X$); here, $X$, $S$, and $R$ are the frequencies of the different
host states as well as their designations. The schematic diagram Figure \ref{Fig.1} leads to the following system of ordinary differential equations
\cite{LB}:
\begin{equation}
\begin{split}
\dfrac{dS}{dt} & =  m\mu+\beta SX-(\tau_{1}+\tau_{2}+\gamma+\mu)S+\sigma\beta cSR,\\
\dfrac{dR}{dt}  &= \beta(1-c)RX-(\mu+\tau_{2}+\gamma)R-\sigma\beta cSR,\\
\dfrac{dX}{dt} & =  (1-m)\mu+(\tau_{1}+\tau_{2}+\gamma)S+(\tau_{2}+\gamma)R-\beta SX-\beta(1-c)RX -\mu X.
\end{split}\label{model}
\end{equation}
The parameters in the model are explained in Table \ref{table1}.
If $\sigma=0$, then the superinfection is ignored and system \eqref{model} is the model studied in \cite{LBL}.

\begin{figure}[h]
\centering \includegraphics[width=0.65\textwidth]{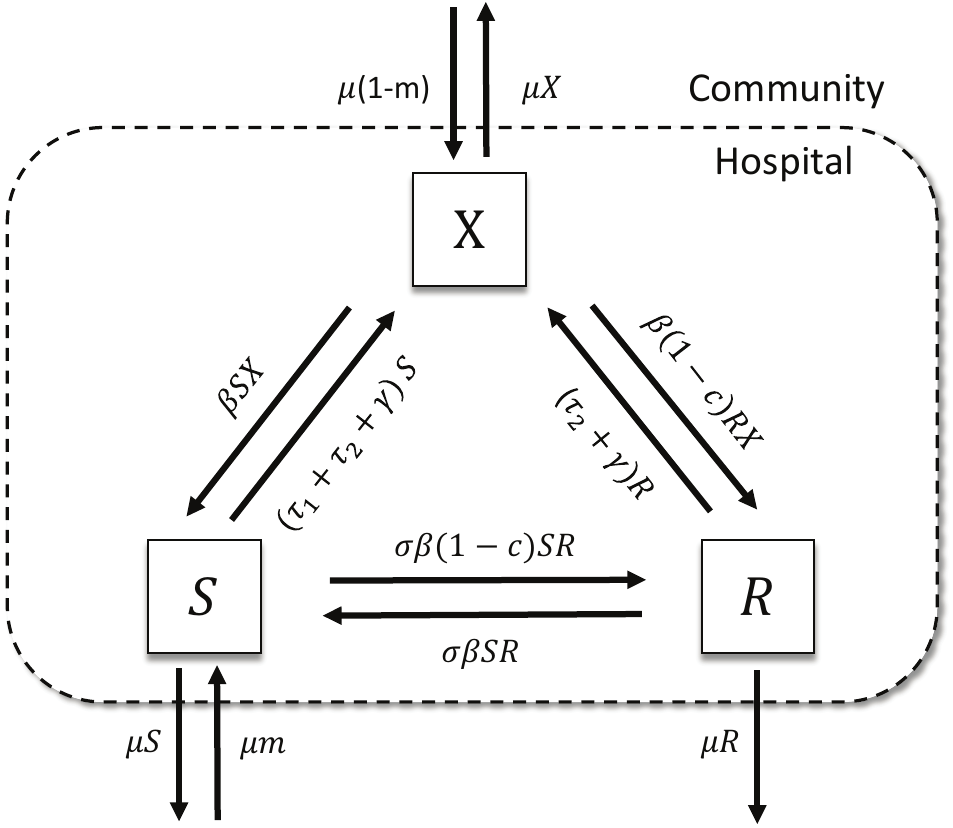} \protect\caption{A compartment model of bacterial transmission dynamics in a hospital
setting. This figure is adopted from \cite{LB}, and we reverse the
arrows from $S$ to $R$. The correction has been confirmed by the
authors of \cite{LB} in the private mail.}
\label{Fig.1}
\end{figure}

\begin{table}[h]
\caption{Description of parameters in the model \eqref{model}}\centering
\doublerulesep=0.4pt
\begin{tabular}{cl}
\hline\hline
\small{Parameters} & \small{Description} \\ \hline
\small{$\mu$} &  \small{Per-capita patient turnover rate, $\mu>0$}   \\
\small{$m$} & \small{Proportion of admitted already colonized with sensitive bacteria, $m\in[0,1]$}\\
\small{$\beta$} & \small{Per-capita primary transmission rate (colonization), $\beta>0$}\\
\small{$\tau_i$} & \small{Per-capita  treatment rate of  drug $i$,  $i=1,2$, $\tau_i\geq0$}\\
\small{$\gamma$} & \small{Per-capita  clearance rate of bacteria due to immune response, $\gamma>0$}\\
\small{$\sigma$} & \small{Relative rate of secondary colonization to that of the primary colonization, $\sigma\in[0,1]$ }\\
\small{$c$} & \small{Fitness ``cost" of a bacterial strain resistant to drug 1, $c\in[0,1)$}\\
\hline\hline
\end{tabular}
\label{table1}
\end{table}

Note that $S+R+X=1$. The model \eqref{model} can be simplified to an equivalent
planar differential system \eqref{plane} by substituting $X=1-S-R$ into the first
two equations of system \eqref{model},
\begin{equation}
\begin{split}
\dfrac{dS}{dt} & = m\mu+\beta S(1-S-R)-(\tau_{1}+\tau_{2}+\gamma+\mu)S+\sigma\beta cSR=U(S,R),\\
\dfrac{dR}{dt} & = R[\beta(1-c)(1-S-R)-(\tau_{2}+\gamma+\mu)-\sigma\beta cS]=V(S,R).
\end{split}\label{plane}
\end{equation}

In this paper we analytically studied system \eqref{plane} instead of system \eqref{model}. We obtained the invasion reproduction number $\mathcal{R}_{ar}$ of antibiotic-resistant bacteria and introduced the connection between $\mathcal{R}_{ar}$ and $\mathcal{R}_{sc}$ and $\mathcal{R}_{rc}$, which are the control reproduction numbers of sensitive bacteria and resistant bacteria. It is of great value in measuring the effectiveness of interventions to reduce resistance in hospitals. One of important findings in our paper is that backward bifurcation may occur when superinfection is included which is impossible in the model without superinfection.
When the backward bifurcation arises, it supports an interesting bistable phenomenon and provides important information on explanation of antibiotic-resistant controls. Moreover, we get the global dynamics for all parameter cases. Numerical simulations help to verify the analytical results obtained. Some measures can be taken to control the transmission of antibiotic-resistant bacteria, for example,
isolation of patients with infections caused by antibiotic-resistant bacteria, properly increased use of drug 2, etc.
Sensitivity analysis of the model based on Latin Hypercube Sampling shows that the treatment of drug 1 ($\tau_1$) has the most important influence on the frequency of patients colonized with antibiotic-resistant bacteria, and they are positively correlated. Thus, it needs to be more cautious to use drug 1.

The rest of this paper is organized as follows.  Section \ref{analysis} is devoted to the studies on
the invasion reproduction number $\mathcal{R}_{ar}$ of antibiotic-resistant
bacteria, the bifurcation analysis and  the global dynamics of the model for different parameter conditions. We give some numerical analysis about the backward bifurcation and dependence on parameters in section \ref{numerical}.  Discussions are provided in section \ref{discussion}. Finally, some detailed proofs and the global dynamics of the model when $m=0$ are given in the Appendix for reference.

\section{Model analysis}\label{analysis}
In what follows we will study system \eqref{plane} to obtain the global dynamics of system \eqref{model}.

Note that $S$ and $R$ are the frequencies of different host states.
It is sufficient to consider the dynamics of system \eqref{plane}
on the following closed region:
\begin{equation}\label{Omega}
\Omega=\{(S,R)|0\leq S\leq1,0\leq R\leq1-S\}.
\end{equation}
Moreover, we have the following proposition on $\Omega$:

\begin{proposition}\label{invar_set}
For system \eqref{plane},  the set $\Omega$, defined in \eqref{Omega}, is positively invariant. That is to say,  all solutions  through a point in $\Omega$ stay in $\Omega$ for all $t\geq 0$.
\end{proposition}
The proof of Proposition \ref{invar_set} is given in the Appendix.

Before the analysis of the model, we introduce two important control reproduction numbers, which is closely related with
the invasion reproduction number and the dynamics of system \eqref{plane}.

In a hypothetical institution where all individuals enter uncolonized ($m=0$), the control reproduction number of
the sensitive bacteria is defined as
\begin{equation}\label{Rsc}
\mathcal{R}_{sc}=\dfrac{\beta}{\tau_1+\tau_2+\gamma+\mu},
\end{equation}
and the control reproduction number of the resistant bacteria is defined as
\begin{equation}\label{Rrc}
\mathcal{R}_{rc}=\dfrac{\beta(1-c)}{\tau_2+\gamma+\mu}.
\end{equation}
They respectively represent the number of secondary infections produced by an infectious individual colonized with the sensitive bacteria
and by an infectious individual colonized with the resistant bacteria during the respective infectious period, when introduced in a population uncolonized.

The global dynamics of system \eqref{plane} when $m>0$ is very different from the case $m=0$. In the present paper, we mainly focus on the case $m>0$ and give the global dynamics of system \eqref{plane} when $m=0$ in the Appendix for reference and comparison.

\subsection{Resistance-free equilibrium and invasion reproduction number}
In this subsection, we study the resistance-free equilibrium of system
\eqref{plane} and get the invasion reproduction number of antibiotic-resistant
bacteria.

It is easy to verify that system \eqref{plane} has a unique boundary
equilibrium, i.e., the resistance-free equilibrium $E_{0}^*(S_{0},0)$,
where $S_{0}$ is the positive root of the equation
\begin{equation}
\beta S^{2}+(\tau_{1}+\tau_{2}+\gamma+\mu-\beta)S-m\mu=0,\label{U0}
\end{equation}
given by
\begin{equation}\begin{split}
S_{0}&=\dfrac{\beta-\tau_{1}-\tau_{2}-\gamma-\mu+\sqrt{(\beta-\tau_{1}-\tau_{2}-\gamma-\mu)^{2}+4\beta m\mu}}{2\beta}\\
&=\dfrac{1}{2}\left[1-\dfrac{1}{\mathcal{R}_{sc}}+\sqrt{\left(1-\dfrac{1}{\mathcal{R}_{sc}}\right)^{2}+\dfrac{4m\mu}{\beta}}\right].\label{S0}
\end{split}\end{equation}
It follows from $m\in(0,1]$ that $S_{0}\in(0,1)$, and thus it is biological feasible. For the case $m>0$, the existence of resistance-free equilibrium does not depend on the scale of $\mathcal{R}_{sc}$, and the disease-free equilibrium no longer exists.

Notice that $S_{0}$ given by \eqref{S0} represents the frequency of individuals
which carry the bacteria sensitive to drug 1 in the absence of antibiotic-resistant infection, while $X_{0}=1-S_{0}$
represents the frequency of individuals who are not colonized with
the sampled species. These two quantities together can be used to describe the invasion reproduction number of antibiotic-resistant
bacteria, which is denoted by $\mathcal{R}_{ar}$ and defined as
\begin{equation}
\mathcal{R}_{ar}= {\displaystyle \frac{\beta(1-c)X_{0}+\sigma\beta(1-c)S_{0}}{\tau_{2}+\gamma+\mu+\sigma\beta S_{0}}}.\label{rate}
\end{equation}
The invasion reproduction number $\mathcal{R}_{ar}$ represents the total number of individuals
colonized with antibiotic-resistant bacteria caused by one ``infected" individual during the entire period of colonization ($1/(\tau_{2}+\gamma+\mu+\sigma\beta S_{0})$) when the antibiotic-resistant bacteria invade into a population where individuals carrying the sensitive bacteria are at the level $S_0$, and individuals not colonized with the sampled species are at the level $X_0$. Note that
the total individuals colonized with antibiotic-resistant bacteria not only include the colonization of ones who are free of the bacterial species of interest, but also the colonization of ones who carry the sensitive bacteria in which case the superinfection occurs.

The invasion reproduction number $\mathcal{R}_{ar}$ also can be written as a function of control reproduction numbers $\mathcal{R}_{rc}$ and $\mathcal{R}_{sc}$ which is implied in the expression of $S_0$:
\begin{equation}\label{Rar}
\mathcal{R}_{ar}= {\displaystyle\frac{(1-c)[1-(1-\sigma)S_{0}]\mathcal{R}_{rc}}{1-c+\sigma S_{0}\mathcal{R}_{rc}}}.
\end{equation}
\begin{remark}\label{Rar&Rrc}
It follows from \eqref{Rar} that $\mathcal{R}_{ar}\geq1\Rightarrow\mathcal{R}_{rc}>1$ and $\mathcal{R}_{rc}\leq1\Rightarrow\mathcal{R}_{ar}<1$.
\end{remark}
The stability of the resistance-free equilibrium $E_{0}^*$ can be determined
by the invasion reproduction number $\mathcal{R}_{ar}$, which is described in the following result.
\begin{theorem}\label{th:E0}
System \eqref{plane} always has a resistance-free equilibrium $E_{0}^*(S_{0},0)$. If $\mathcal{R}_{ar}<1$, then $E_{0}^*$ is a locally asymptotically stable (l.a.s.) node; if $\mathcal{R}_{ar}>1$, then $E_{0}^*$ is an unstable saddle; if $\mathcal{R}_{ar}=1$, then $E_0^*$ is l.a.s. when  $f(\mathcal{R}_{rc})\leq0$, and is unstable when $f(\mathcal{R}_{rc})>0$, where
\begin{equation}\label{f}
f(\mathcal{R}_{rc})=[c\sqrt{\beta(1-c)\sigma(1-\sigma)}-(1-c+\sigma c)\sqrt{m\mu}]\mathcal{R}_{rc}
-c\sqrt{\beta(1-c)\sigma(1-\sigma)}.
\end{equation}
\end{theorem}
The proof of Theorem \ref{th:E0} can be found in the Appendix.
\begin{remark}
In the case of $\mathcal{R}_{ar}=1$, the condition $f(\mathcal{R}_{rc})>0$ can be replaced by $\beta>\beta^*$ and $\mathcal{R}_{rc}>\mathcal{R}_{rc}^*$, thus, under the new conditions, $E_0^*$ is unstable, otherwise $E_0^*$ is l.a.s., where
\begin{equation}\begin{split}\label{Rrc*}
\beta^*&=\dfrac{m\mu(1-c+\sigma c)^2}{\sigma(1-c)(1-\sigma)c^2},\\
\mathcal{R}_{rc}^*&=\dfrac{c\sqrt{\beta(1-c)\sigma(1-\sigma)}}{c\sqrt{\beta(1-c)\sigma(1-\sigma)}-(1-c+\sigma c)\sqrt{m\mu}}.
\end{split}\end{equation}
\end{remark}
\subsection{Existence and stability analysis of interior equilibria}

This subsection is devoted to the study of the existence and stability of interior
equilibria.

An equilibrium $E^{*}(S^{*},R^{*})$ is called an \textit{interior
equilibrium} of system \eqref{plane} if $E^{*}(S^{*},R^{*})$ is
an interior point of $\Omega$, i.e.,
\begin{equation}
0<S^{*}<1,\,\,0<R^{*}<1-S^{*}.\label{inter}
\end{equation}

\subsubsection{Existence of interior equilibria}

It follows from $V(S^{*},R^{*})=0$ that
\begin{equation}
R^{*}=-\frac{1-c+\sigma c}{1-c}S^{*}+\frac{\mathcal{R}_{rc}-1}{\mathcal{R}_{rc}}.\label{R*}
\end{equation}
Substituting it into $U(S^{*},R^{*})=0$, we get the following equation
satisfied by the abscissa of the interior equilibrium (for which we
drop ``$*$'' for simplicity)
\begin{equation}
\Phi(S)=\phi_{2}S^{2}+\phi_{1}S+\phi_{0}=0,\label{equ:S}
\end{equation}
where
\begin{equation}\label{phi}
\phi_{2}=c^{2}\sigma(1-\sigma),\quad\phi_{1}=(1-c)\left(\dfrac{1-c\sigma}{\mathcal{R}_{rc}}+c\sigma-
\dfrac{1}{\mathcal{R}_{sc}}\right),\quad\phi_{0}=\dfrac{(1-c)m\mu}{\beta}.
\end{equation}

To make sure that $E^{*}(S^{*},R^{*})$ is an interior equilibrium of system \eqref{plane},
by \eqref{inter} and \eqref{R*}, we have the following conditions which must be satisfied:
\begin{equation}\label{a^*}
\mathcal{R}_{rc}>1 \quad \hbox{and} \quad 0<S^{*}<\frac{(1-c)(\mathcal{R}_{rc}-1)}{(1-c+c\sigma)\mathcal{R}_{rc}}=a^{*}<1.
\end{equation}
Note that $a^*<1$ holds naturally and $a^*>0$ if and only if $\mathcal{R}_{rc}>1$. If $\mathcal{R}_{rc}\leq1$, then $\mathcal{R}_{ar}<1$ by Remark \ref{Rar&Rrc} and system \eqref{plane} has no interior equilibrium. Hence, the number of roots of equation \eqref{equ:S} in the interval $(0,a^*)$ corresponds to the number of interior
equilibria of system \eqref{plane}, and it is no more than 2. Although the coefficients of $\Phi(S)$ are a little complex,  we can determine the number of solutions $S^*$ by examining the properties of the function $\Phi(S)$.

A direct computation leads to
\begin{equation}
\Phi(0)=\phi_{0}>\text{0},\quad\Phi(a^{*})=\Phi_{0}(a^{*})(1-\mathcal{R}_{ar}),\label{boundary}
\end{equation}
where
\[
\Phi_{0}(a^{*})=\dfrac{(1-c)(1-c+\sigma S_{0}\mathcal{R}_{rc})}{(1-c+\sigma c)\mathcal{R}_{rc}}\left(a^{*}-\tilde{{S_{0}}}\right)>0,
\]
with $\tilde{S_{0}}$ being the other root of the equation \eqref{U0} and $\tilde{S_{0}}<0$.

If $c\sigma(1-\sigma)=0$, then $\phi_2=0$ and $\Phi(S)$ is linear, and thus
system \eqref{plane} has a unique interior equilibrium if and only if $\mathcal{R}_{ar}>1$ (i.e., $\Phi(a^{*})<0$).

If $c\sigma(1-\sigma)>0$, then $\Phi(S)$ is quadratic and it has at most 2 zeros. (i) When $\mathcal{R}_{ar}>1$, we have $\mathcal{R}_{rc}>1$ and $\Phi(a^{*})<0$.  It follows from $\Phi(0)>0$ that the system has a unique interior equilibrium.
(ii) When $\mathcal{R}_{ar}<1$ (in this case $\mathcal{R}_{rc}>1$ can't be derivated), we suppose $\mathcal{R}_{rc}>1$ and get $\Phi(a^{*})>0$. The system \eqref{plane}
may have 0, 1, or 2 interior equilibria, which is determined by the sign of the
discriminant
\begin{equation}\begin{split}
\Delta&=\phi_{1}^{2}-4\phi_{0}\phi_{2}\\
&=(1-c)^2\left(\dfrac{1-c\sigma}{\mathcal{R}_{rc}}+c\sigma-\dfrac{1}{\mathcal{R}_{sc}}\right)^{2}-4c^{2}\sigma(1-\sigma)(1-c)\dfrac{m\mu}{\beta}\label{discriminant}
\end{split}\end{equation}
and the position of symmetry axis $-\phi_{1}/(2\phi_{2})$ of quadratic polynomial equation \eqref{equ:S}. The detailed
analysis is as follows:

Note that $\Phi(S)$ achieves its minimum at $-\phi_{1}/(2\phi_{2})$. Hence, system \eqref{plane} has two different interior equilibria if and only if
\begin{equation}\label{de0}
\mathcal{R}_{ar}<1,\quad\Delta>0,\quad0<-\frac{\phi_{1}}{2\phi_{2}}<a^{*}.
\end{equation}
If
\begin{equation}\label{de1}
\Delta=0,\quad0<-\frac{\phi_{1}}{2\phi_{2}}<a^{*},
\end{equation}
then system \eqref{plane} has only one interior equilibrium.

(iii) When $\mathcal{R}_{ar}=1$, we have $\mathcal{R}_{rc}>1$ and  $\Phi(a^{*})=0$. It suffices to ensure  $-\phi_{1}/(2\phi_{2})<a^*$ that system \eqref{plane} can have a unique interior equilibrium. This is because that $\Phi(S)$\ assumes its minimum at $-\phi_{1}/(2\phi_{2})$ and it follows from $-\phi_{1}/(2\phi_{2})<a^*$ and $\Phi(a^{*})=0$ that this minimum is less than 0. Combining $\Phi(0)>0$, we have $-\phi_{1}/(2\phi_{2})>0$ and $\Phi(S)$ has exactly one zero in $(0,a^{*})$.

System \eqref{plane} has none interior equilibria for all other cases.

We simplify the existence conditions \eqref{de0} and \eqref{de1}, and summarize the discussions  above as follows:
\begin{theorem}\label{th:S}
Let $\mathcal{R}_{ar}$ be the invasion reproduction number defined in \eqref{rate} or \eqref{Rar}, and $\mathcal{R}_{sc}$, $\mathcal{R}_{rc}$ be the control reproduction numbers defined in \eqref{Rsc} and \eqref{Rrc}, respectively.
\begin{itemize}
\item[(a)]  If $\mathcal{R}_{ar}>1$, then system \eqref{plane} has a unique interior equilibrium;
\item[(b)]  if $c\sigma(1-\sigma)>0$, $\beta>\beta^*$ and $\mathcal{R}_{rc}>\mathcal{R}_{rc}^*$, then there exists
    a new threshold $\mathcal{R}_{ar}^c$ ($<1$), such that
    \begin{itemize}
\item[(i)]  when $\mathcal{R}_{ar}^c<\mathcal{R}_{ar}<1$ ($\mathcal{R}_{sc}^{1}<\mathcal{R}_{sc}<\mathcal{R}_{sc}^{\triangle}$), system \eqref{plane} has two different interior equilibria;
\item[(ii)] when $\mathcal{R}_{ar}=\mathcal{R}_{ar}^c$ ($\mathcal{R}_{sc}=\mathcal{R}_{sc}^{\triangle}$), system \eqref{plane} has a unique interior equilibrium;
\item[(iii)] when $\mathcal{R}_{ar}=1$ ($\mathcal{R}_{sc}=\mathcal{R}_{sc}^{1}$), system \eqref{plane} has a unique interior equilibrium;
    \end{itemize}
\item[(c)]  system \eqref{plane} has none interior equilibria for all other cases,
\end{itemize}
where $\beta^*$ and $\mathcal{R}_{rc}^*$ are given by \eqref{Rrc*}, and
\begin{equation}\label{Rarc}
\mathcal{R}_{ar}^c={\displaystyle \frac{\beta(1-c)(1-(1-\sigma)S_0^{\triangle})}{\tau_{2}+\gamma+\mu+\sigma\beta S_0^{\triangle} }},
\end{equation}
or equivalently,
\begin{equation}\label{Rarc1}
\mathcal{R}_{ar}^c= {\displaystyle\frac{(1-c)[1-(1-\sigma)S_0^{\triangle}]\mathcal{R}_{rc}}{1-c+\sigma S_0^{\triangle}\mathcal{R}_{rc}}},
\end{equation}
with
\begin{equation*}\begin{split}
S_0^{\triangle}=
\dfrac{1}{2}\left[1-\dfrac{1}{\mathcal{R}_{sc}^{\triangle}}+\sqrt{\left(1-\dfrac{1}{\mathcal{R}_{sc}^{\triangle}}\right)^{2}+\dfrac{4 m\mu}{\beta}}\right].\label{S_0^c}
\end{split}\end{equation*}
Here,
\begin{equation}
\begin{split}\label{Rsc1}
\mathcal{R}_{sc}^{\triangle}&=\dfrac{\beta(1-c)\mathcal{R}_{rc}}{2c\mathcal{R}_{rc}
\sqrt{\beta\sigma(1-c)(1-\sigma)m\mu}+\beta(1-c)(1+c\sigma(\mathcal{R}_{rc}-1))},\\
\mathcal{R}_{sc}^{1}&=\dfrac{\beta(1-c)(1-c+c\sigma)\mathcal{R}_{rc}(\mathcal{R}_{rc}-1)}
{m\mu(1-c+c\sigma)^2\mathcal{R}_{rc}^2+\beta(1-c)(\mathcal{R}_{rc}-1)(1-c+c\sigma\mathcal{R}_{rc})}\\
\end{split}
\end{equation}
are derived from the conditions $\Delta=0$ and $\mathcal{R}_{ar}=1$, respectively.
\end{theorem}
\begin{remark}
(i) The condition \eqref{de1} actually determines a lower bound $\mathcal{R}_{ar}^c$ of $\mathcal{R}_{ar}$ (note that in this case, $\Phi(a^*)>\Phi(-\phi_{1}/(2\phi_{2}))=0$, thus condition \eqref{de1} implies that $\mathcal{R}_{ar}^c<1$ by \eqref{boundary}), such that system \eqref{plane} has two different interior equilibria if and only if $\mathcal{R}_{ar}^c<\mathcal{R}_{ar}<1$.

(ii) The condition \eqref{de1} is equivalent to $\mathcal{R}_{sc}=\mathcal{R}_{sc}^{\triangle}>\mathcal{R}_{sc}^{a^*}$. However, \[
\mathcal{R}_{sc}^{\triangle}-\mathcal{R}_{sc}^{a^*}=
\dfrac{2c\sqrt{\beta(1-c)(1-\sigma)\sigma}\mathcal{R}_{sc}^{\triangle}\mathcal{R}_{sc}^{a^*}f(\mathcal{R}_{rc})}{\beta(1-c)(1-c+\sigma c)\mathcal{R}_{rc}}>0
\]
requires that $f(\mathcal{R}_{rc})>0$, i.e., $\beta>\beta^*$ and $\mathcal{R}_{rc}>\mathcal{R}_{rc}^*$, where $f(\mathcal{R}_{rc})$, $\beta^*$ and $\mathcal{R}_{rc}^*$ are given in \eqref{f} and \eqref{Rrc*}, respectively, and
\[
\mathcal{R}_{sc}^{a^*}=\dfrac{(1-c+c\sigma)\mathcal{R}_{rc}}{1-c+c\sigma\mathcal{R}_{rc}+c^2\sigma(1-\sigma)
(\mathcal{R}_{rc}-1)}
\]
comes from $-\phi_{1}/(2\phi_{2})=a^*$. This provides the existence condition of $\mathcal{R}_{sc}^{\triangle}$ and $\mathcal{R}_{ar}^c$.

(iii) Similarly, under the hypothesis of $f(\mathcal{R}_{rc})>0$, condition \eqref{de0} is equivalent to $\mathcal{R}_{sc}^{1}<\mathcal{R}_{sc}<\mathcal{R}_{sc}^{\triangle}$, i.e., $\mathcal{R}_{ar}^c<\mathcal{R}_{ar}<1$.

\end{remark}

\begin{remark}
The existence of interior equilibria depends not only on $\mathcal{R}_{ar}$, but also on $\mathcal{R}_{rc}$ and $\mathcal{R}_{sc}$. To see how the number of interior equilibria varies, we choose $\mathcal{R}_{rc}$ and $\mathcal{R}_{sc}$ as the bifurcation parameters and consider bifurcations in $(\mathcal{R}_{rc}, \mathcal{R}_{sc})$ plane.

For the case of $c\sigma(1-\sigma)>0$ and $\beta>\beta^*$,
\begin{equation*}
\mathcal{R}_{sc}^{\triangle}-\mathcal{R}_{sc}^1
=\dfrac{\mathcal{R}_{sc}^{\triangle}\mathcal{R}_{sc}^1f(\mathcal{R}_{rc})^2}{\beta(1-c)(1-c+c\sigma)\mathcal{R}_{rc}
(\mathcal{R}_{rc}-1)}\geq0,
\end{equation*}
thus, the curves $\mathcal{R}_{sc}=\mathcal{R}_{sc}^{\triangle}$ and $\mathcal{R}_{sc}=\mathcal{R}_{sc}^{1}$ contact at one point only when $\mathcal{R}_{rc}=\mathcal{R}_{rc}^*$, and thus the $(\mathcal{R}_{rc}, \mathcal{R}_{sc})$ plane will be divided into three regions $\mathcal D_i$, $i=0,1,2$; for all other cases, the $(\mathcal{R}_{rc}, \mathcal{R}_{sc})$ plane only be divided into two regions $\mathcal D_i$, $i=0,1$. If $(\mathcal{R}_{rc}, \mathcal{R}_{sc})\in\mathcal D_i$, then system \eqref{plane} has $i$ interior equilibria, see Figure \ref{Fig.2}.

\begin{figure}[h]
\centering
\subfigure[]{
\label{Fig2.sub.1}
\includegraphics[width=0.48\textwidth]{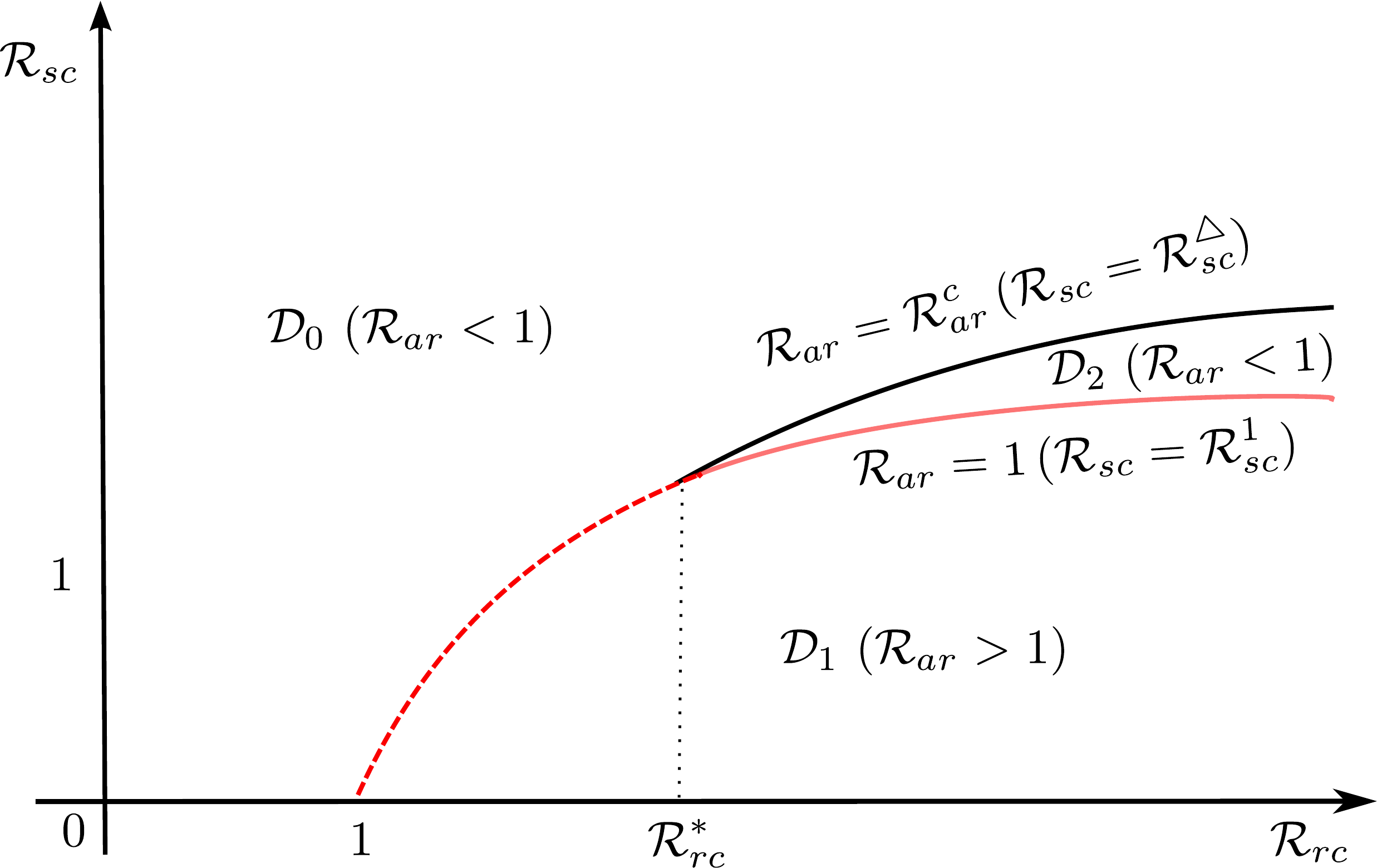}}
\subfigure[]{
\label{Fig2.sub.2}
\includegraphics[width=0.48\textwidth]{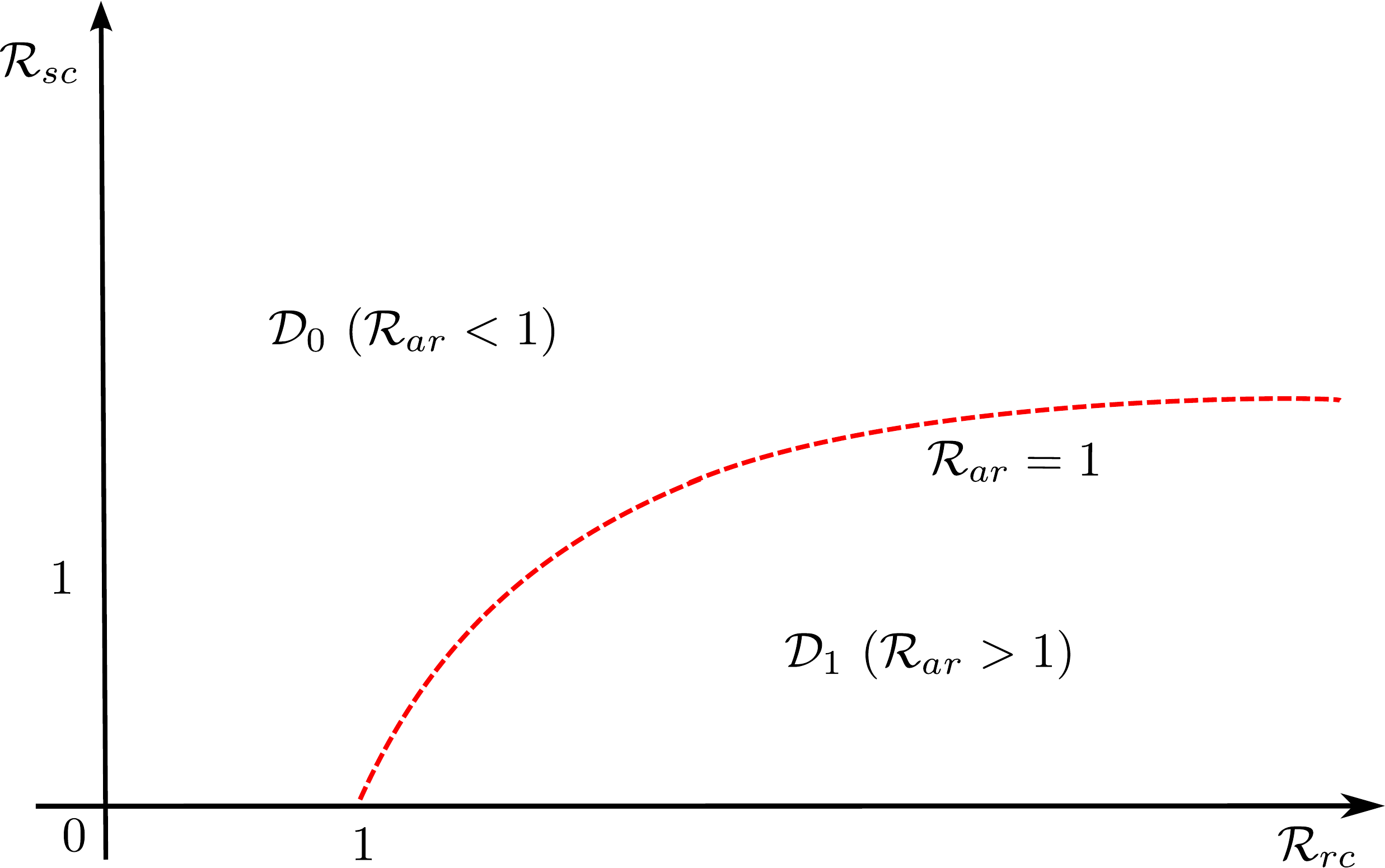}}
\caption{The existence of interior equilibria. System \eqref{plane} has $i$ interior equilibria if $(\mathcal{R}_{rc},\mathcal{R}_{sc})\in \mathcal D_i,\,\,i=0,1,2$. Plot(a) shows the bifurcations when $c\sigma(1-\sigma)>0$ and $\beta>\beta^*$. The red curve plots $\mathcal{R}_{sc}=\mathcal{R}_{sc}^{1}$, i.e., $\mathcal{R}_{ar}=1$, and the black curve plots $\mathcal{R}_{sc}=\mathcal{R}_{sc}^{\triangle}$, i.e., $\mathcal{R}_{ar}=\mathcal{R}_{ar}^c$. The dashed part represents there is no interior equilibrium, and the solid part corresponds to one interior equilibrium. For all other parameters conditions, the bifurcation diagram can be shown by Plot(b).}
\label{Fig.2}
\end{figure}
\end{remark}

\subsubsection{Local stability of interior equilibria}

The Jacobian matrix of system \eqref{plane} at the interior equilibrium $E^{*}(S^{*},R^{*})$
is given by
\[
J(E^{*})=\left(\begin{array}{cc}
U_S(S^*,R^*) & -\beta(1-c\sigma)S^{*}\\
-\beta(1-c+c\sigma)R^{*} & V_R(S^*,R^*)
\end{array}\right),
\]
where
\begin{equation*}\begin{split}
U_S(S^*,R^*)&=\dfrac{\partial U(S^*,R^*)}{\partial S}=\beta-\tau_{1}-\tau_{2}-\gamma-\mu-2\beta S^{*}-\beta(1-c\sigma)R^{*},\\
V_R(S^*,R^*)&=\dfrac{\partial V(S^*,R^*)}{\partial R}=\beta(1-c)-\tau_{2}-\gamma-\mu-\beta(1-c+c\sigma)S^{*}-2\beta(1-c)R^{*},
\end{split}\end{equation*}
or the simplified matrix
\[
J(E^{*})=\left(\begin{array}{cc}
-\beta S^{*}-m\mu/S^{*} & -\beta(1-c\sigma)S^{*}\\
-\beta(1-c+c\sigma)R^{*} & -\beta(1-c)R^{*}
\end{array}\right).
\]
Obviously, the trace and the determinant of $J(E^{*})$ are:
\begin{equation}\label{det}
\begin{split}
\mathrm{tr}J(E^{*}) & =  -\beta S^{*}-m\mu/S^{*}-\beta(1-c)R^{*}<0,\\
\mathrm{det}J(E^{*}) & =  \beta R^{*}[m\mu(1-c)-\beta c^{2}\sigma(1-\sigma)S^{*2}]/S^{*}=\beta R^{*}(\phi_{0}-\phi_{2}S^{*2})/S^{*}.
\end{split}
\end{equation}
Note that
\begin{equation}
\begin{split}\label{delta}
&(\mathrm{tr}J(E^{*}))^2-4\mathrm{det}J(E^{*})\\
=&(\beta S^{*}+m\mu/S^{*}-\beta(1-c)R^{*})^2+4\beta^2(1-c\sigma)(1-c+c\sigma)S^{*}R^*>0.
\end{split}
\end{equation}
Thus, the stability and the type of the interior equilibria can be determined by \eqref{det} and \eqref{delta}, which is
described in more detail in the following result.

\begin{theorem}\label{th:local}
Consider the interior equilibria described in Theorem \ref{th:S}.
\begin{itemize}
\item[(a)]  If $\mathcal{R}_{ar}>1$, then the unique interior equilibrium is a l.a.s. node;
\item[(b)]  if $c\sigma(1-\sigma)>0$, $\beta>\beta^*$ and $\mathcal{R}_{rc}>\mathcal{R}_{rc}^*$, then for the new threshold $\mathcal{R}_{ar}^c$ ($<1$) defined in \eqref{Rarc} or \eqref{Rarc1},
    \begin{itemize}
\item[(i)]  when $\mathcal{R}_{ar}^c<\mathcal{R}_{ar}<1$ ($\mathcal{R}_{sc}^{1}<\mathcal{R}_{sc}<\mathcal{R}_{sc}^{\triangle}$), one interior equilibrium is a l.a.s. node, and the other is an unstable saddle;
\item[(ii)] when $\mathcal{R}_{ar}=\mathcal{R}_{ar}^c$ ($\mathcal{R}_{sc}=\mathcal{R}_{sc}^{\triangle}$), the unique interior equilibrium is an unstable saddle-node;
\item[(iii)] when $\mathcal{R}_{ar}=1$ ($\mathcal{R}_{sc}=\mathcal{R}_{sc}^{1}$), the unique interior equilibrium is a l.a.s. node.
    \end{itemize}
\end{itemize}
\end{theorem}
The proof of Theorem \ref{th:local} can be found in the Appendix.
\begin{remark}
From the analysis of the case $c\sigma(1-\sigma)=0$, we know that system \eqref{plane} can not have a backward bifurcation when superinfection is not considered ($\sigma=0$). However, system \eqref{plane} which includes the superinfection ($\sigma>0$) may exhibit the backward bifurcation, see Theorem \ref{th:local}. Apparently, the occurrence of backward
bifurcation will make it more challenging to control the antibiotic-resistant infection since the threshold for eradication is reduced from 1 to $\mathcal{R}_{ar}^c$. The determination of the new threshold plays a key role in providing quantitative criteria to evaluate the success of interventions to control infections in hospitals.
\end{remark}

\subsection{Global dynamics}\label{global}

Based on the analysis in the previous subsections, we are going to give
global dynamics of system \eqref{plane} on the positively invariant set
$\Omega$. To do this, we use Dulac theorem \cite{Y} to rule out the existence of limit cycles and
homoclinic loops.

\begin{proposition}\label{global}
System \eqref{plane} has no limit cycles or homoclinic loops.
\end{proposition}

\begin{proof}
Introduce a function
\[
B(S,R)=\dfrac{1}{SR},
\]
then
\[
\dfrac{\partial (BU)}{\partial S}+\dfrac{\partial (BV)}{\partial R}=-\dfrac{m\mu}{RS^2}-\dfrac{\beta}{R}-
\dfrac{\beta(1-c)}{S}<0,
\]
for all $S>0$, $R>0$.

The proposition follows from Dulac theorem \cite{Y}.
\end{proof}

By Proposition \ref{invar_set}, Theorems \ref{th:E0}, \ref{th:S}, \ref{th:local} and Proposition \ref{global}, we obtain the global dynamics of system \eqref{plane} as follows.

\begin{theorem}\label{th} Let $\mathcal{R}_{ar}$ be the invasion reproduction number defined in \eqref{rate} or \eqref{Rar}, and $\mathcal{R}_{sc}$, $\mathcal{R}_{rc}$ be the control reproduction numbers defined in \eqref{Rsc} and \eqref{Rrc}, respectively. The following statements hold:
\begin{itemize}
\item[(a)]  if $\mathcal{R}_{ar}>1$, then the phase portrait of system \eqref{plane} is topologically equivalent to Figure \ref{Fig3.sub.1};
\item[(b)]  if $c\sigma(1-\sigma)>0$, $\beta>\beta^*$, and $\mathcal{R}_{rc}>\mathcal{R}_{rc}^*$, then there exists a new threshold $\mathcal{R}_{ar}^c$ ($<1$) defined in \eqref{Rarc} or \eqref{Rarc1}, such that
    \begin{itemize}
\item[(i)]  when $\mathcal{R}_{ar}^c<\mathcal{R}_{ar}<1$ ($\mathcal{R}_{sc}^{1}<\mathcal{R}_{sc}<\mathcal{R}_{sc}^{\triangle}$), the phase portrait of system \eqref{plane} is topologically equivalent to Figure \ref{Fig3.sub.2};
\item[(ii)] when $\mathcal{R}_{ar}=\mathcal{R}_{ar}^c$ ($\mathcal{R}_{sc}=\mathcal{R}_{sc}^{\triangle}$), the phase portrait of system \eqref{plane} is topologically equivalent to Figure \ref{Fig3.sub.3};
\item[(iii)] when $\mathcal{R}_{ar}=1$ ($\mathcal{R}_{sc}=\mathcal{R}_{sc}^{1}$), the phase portrait of system \eqref{plane} is topologically equivalent to Figure \ref{Fig3.sub.1};
    \end{itemize}
\item[(c)] for all other cases, the phase portrait of system \eqref{plane} is topologically equivalent to Figure \ref{Fig3.sub.4},
\end{itemize}
where $\beta^*$ and $\mathcal{R}_{rc}^*$ are given by \eqref{Rrc*}, and $\mathcal{R}_{sc}^1$ and $\mathcal{R}_{sc}^{\triangle}$ are given by \eqref{Rsc1}.
\end{theorem}

\begin{figure}[h]
\centering
\subfigure[]{
\label{Fig3.sub.1}
\includegraphics[width=0.39\textwidth]{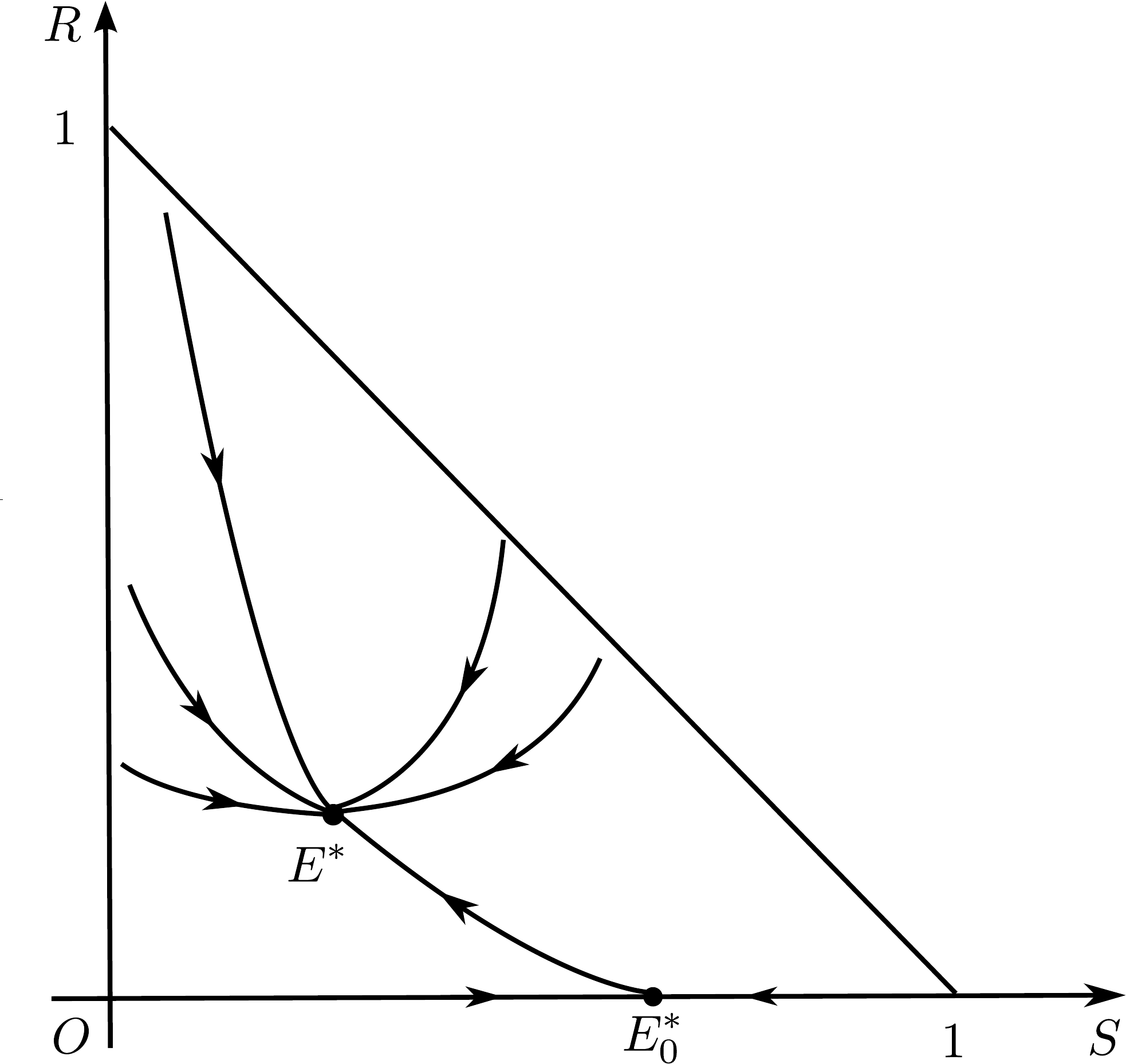}}
\subfigure[]{
\label{Fig3.sub.2}
\includegraphics[width=0.39\textwidth]{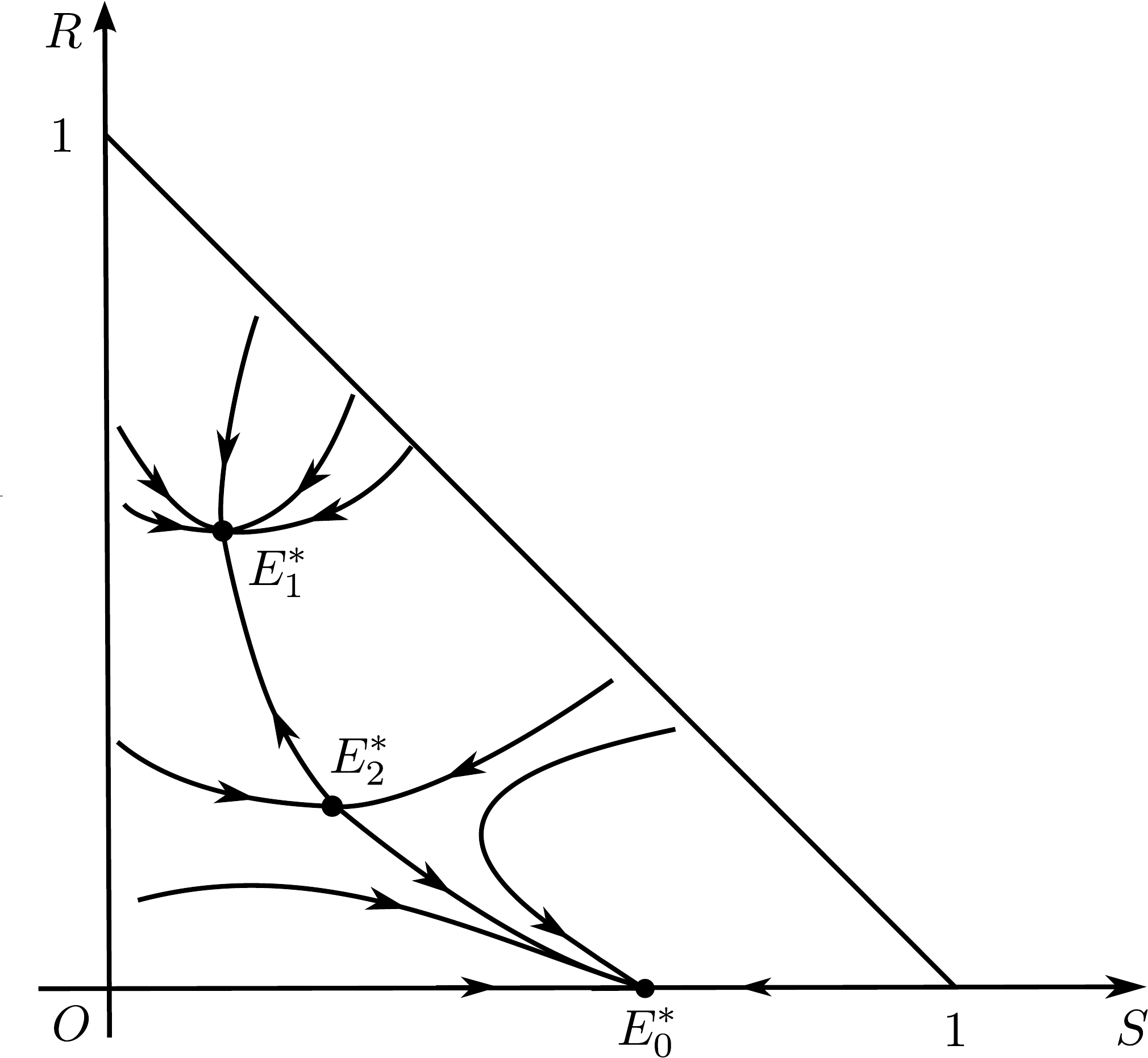}}
\subfigure[]{
\label{Fig3.sub.3}
\includegraphics[width=0.39\textwidth]{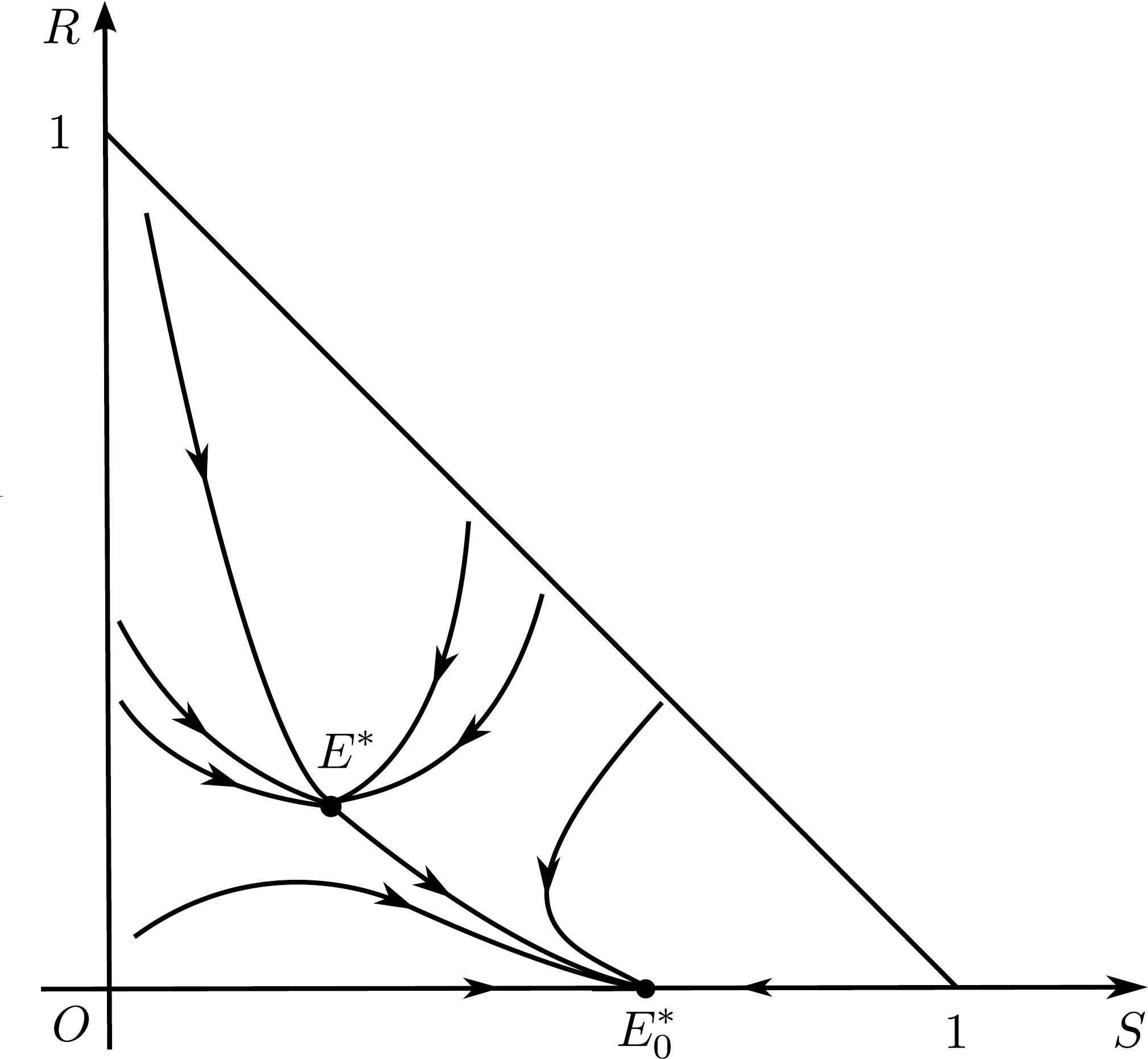}}
\subfigure[]{
\label{Fig3.sub.4}
\includegraphics[width=0.39\textwidth]{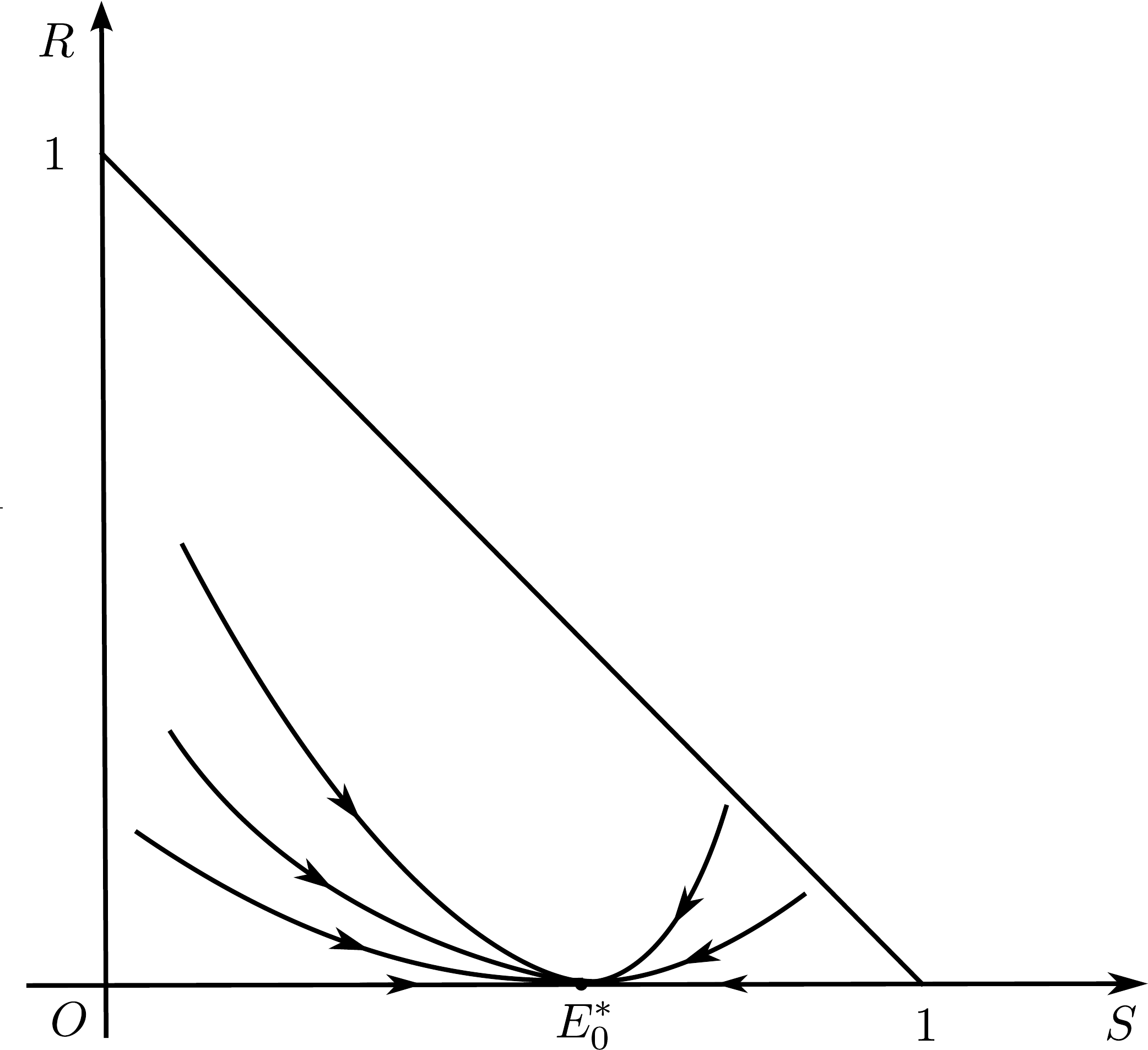}}
\caption{The phase portraits of system \eqref{plane}. }
\label{Fig.3}
\end{figure}

\section{Numerical analysis}\label{numerical}

This section is devoted to numerical simulations for the global results we obtained in last section, the possible backward bifurcation, and the dependence of $\mathcal{R}_{ar}^c$ given in \eqref{Rarc} on parameters.

\begin{figure}[h]
\centering
\includegraphics[width=0.45\textwidth]{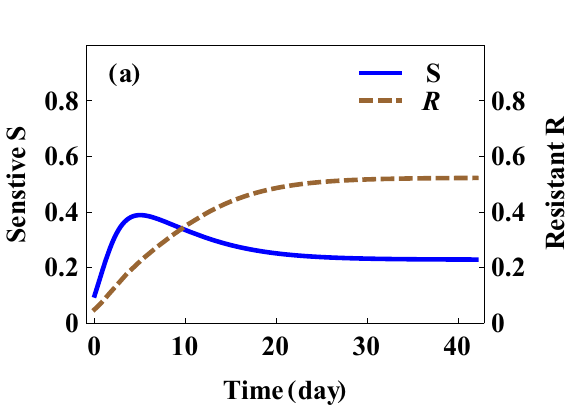}\quad
\includegraphics[width=0.45\textwidth]{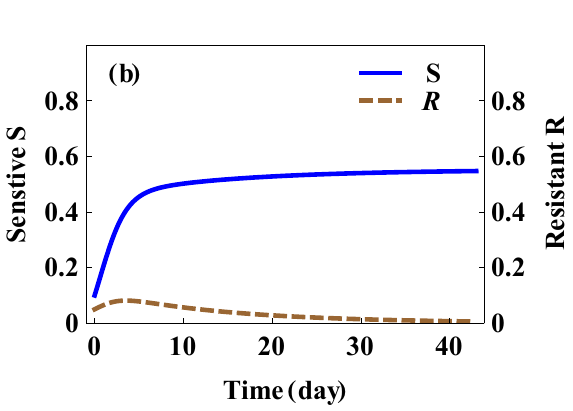}
\caption{Simulation results for the cases $\mathcal{R}_{ar}>1$ (a), and $\mathcal{R}_{ar}<1$ and no interior equilibrium exists (b).}
\label{portrait}
\end{figure}
\begin{figure}[h]
\centering
\includegraphics[width=0.6\textwidth]{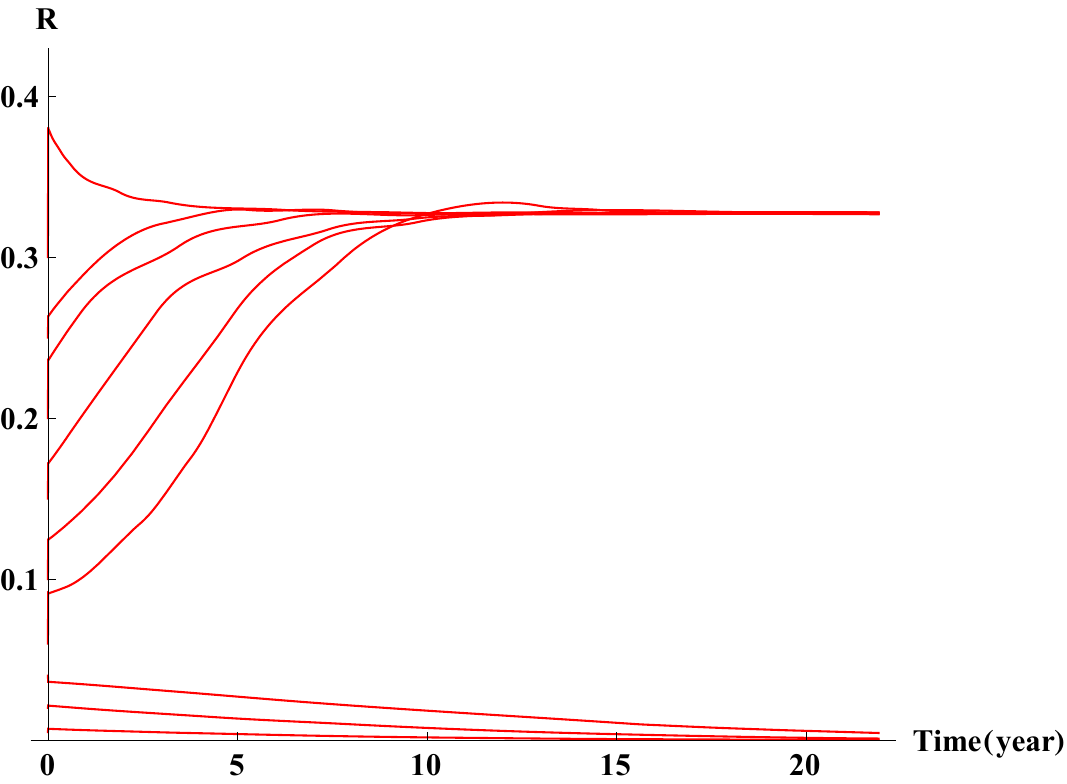}
\caption{Simulation of system \eqref{plane} for the case $\mathcal{R}_{ar}^c<\mathcal{R}_{ar}<1$. In this case, system \eqref{plane} has two different interior equilibria. Solution curves for different initial conditions are plotted, only demonstrating the fraction of individuals $R$ which carry the bacteria resistant to drug 1.}
\label{bistable}
\end{figure}
Figure \ref{portrait} demonstrates the simulation results of the system for the case $\mathcal{R}_{ar}>1$ (a), and the case $\mathcal{R}_{ar}<1$ and no interior equilibrium exists (b). Numerical results show that, in (a), system \eqref{plane} has a unique interior equilibrium,  and all the solution curves with initial values $S(0)\geq0,\, R(0)>0$ will eventually converge to this interior equilibrium; in (b), the system will stabilize at the resistance-free equilibrium, which is globally asymptotically stable. The parameter values in Figure \ref{portrait}(a) are
$m=0.75$, $\mu=0.1$, $\beta=1$,
$\tau_1=0.35$, $\tau_2=0.1$, $\gamma=1/30$, $\sigma=0.25$, $c=0.05$, and thus $\mathcal{R}_{ar}=1.5$; most parameter values in Figure \ref{portrait}(b) are the same as those in Figure \ref{portrait}(a), except $\tau_1=0.1$ and $\tau_2=0.35$, and in this case $\mathcal{R}_{ar}=0.9$. The parameter values we used come from \cite{C,LBL,LB}.

\begin{figure}[h]
\centering
\includegraphics[width=0.6\textwidth]{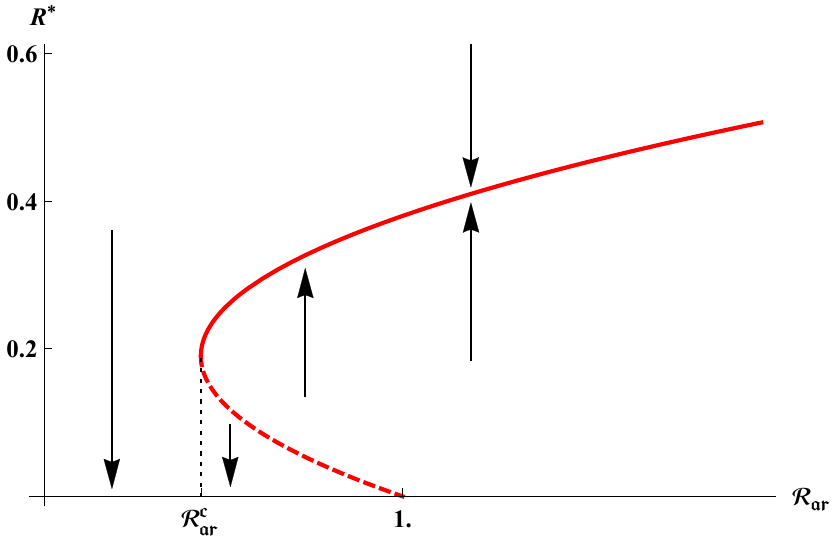}
\caption{Illustration of the backward bifurcation described in Theorem \ref{th:local}.}
\label{bifurcation}
\end{figure}

Figure \ref{bistable} plots the solution curves of system \eqref{plane} for different initial conditions when $\mathcal{R}_{ar}^c<\mathcal{R}_{ar}<1$. Only the frequency of patients with the antibiotic-resistant bacteria $R$ are shown.
We can see clearly that the solution curves with the initial values near $R(0)=0$ will converge to the resistance-free equilibrium $E_0$, while the solution curves with the larger initial values will converge to the stable interior equilibrium $E_1^*$ as $t\rightarrow\infty$, which implies a bistable phenomenon.
The parameter values in this figure are:
$m=0.2$, $\mu=1/15$, $\beta=5$,
$\tau_1=0.3$, $\tau_2=0.19$, $\gamma=1/30$, $\sigma=0.3$, $c=0.15$. For this set of parameter values, $\mathcal{R}_{ar}=0.9997$ and $\mathcal{R}_{ar}^c=0.9994$.

Figure \ref{bifurcation} shows the backward bifurcation described in Theorem \ref{th:local}.  The equilibrium fraction of individuals carrying the antibiotic-resistant bacteria $R^*$ is plotted as a function of $\mathcal{R}_{ar}$. The solid part and the dashed part represent the stable and unstable interior equilibria, respectively. The leftmost point of the curve (the intersection of the solid and dashed branches) corresponds to the lower bound of $\mathcal{R}_{ar}$, $\mathcal{R}_{ar}^c$, for the existence of two interior equilibria. The vertical arrows indicate the convergence of solutions with initial
conditions in the respective regions as $t\rightarrow\infty$. In this figure, we choose $\tau_1$ as the independent variable to vary $\mathcal{R}_{ar}$, and other parameter values are the same as those in Figure \ref{bistable}.

\begin{figure}[h]
\centering
\includegraphics[width=0.45\textwidth]{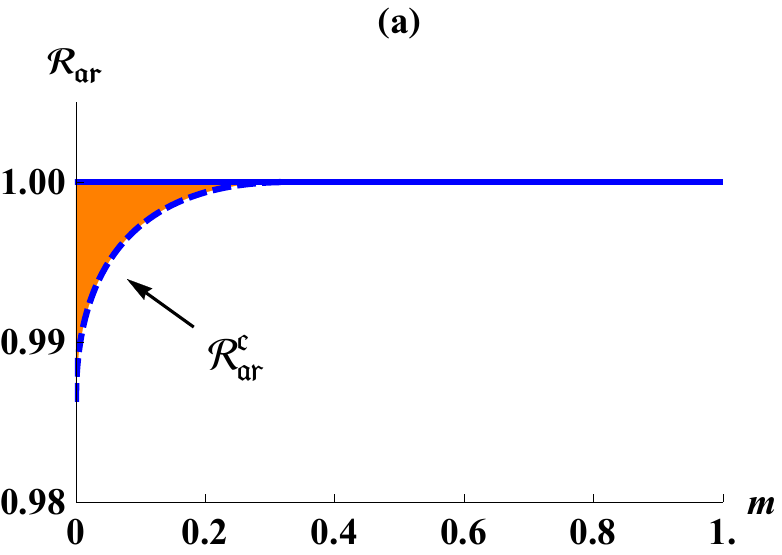}\quad
\includegraphics[width=0.45\textwidth]{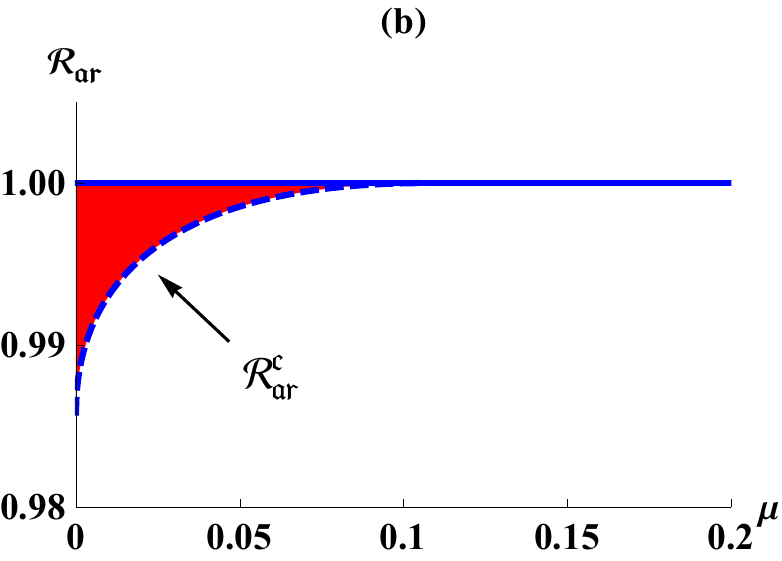}
\includegraphics[width=0.45\textwidth]{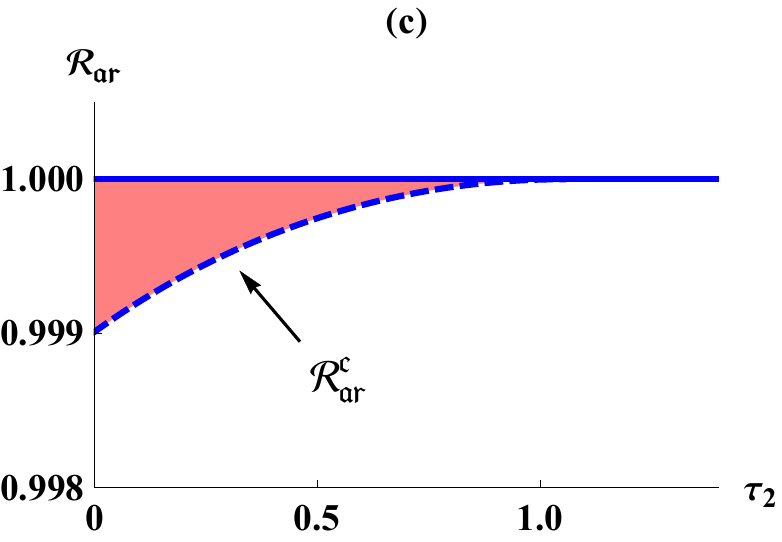}\quad
\includegraphics[width=0.45\textwidth]{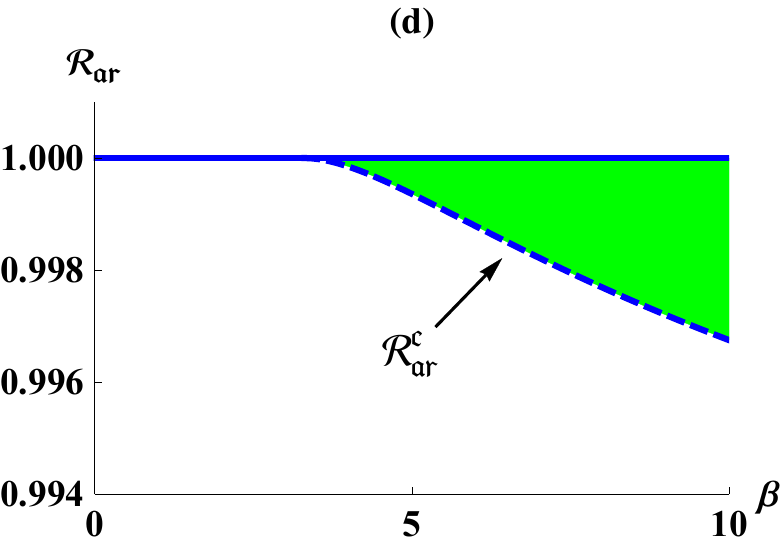}
\includegraphics[width=0.45\textwidth]{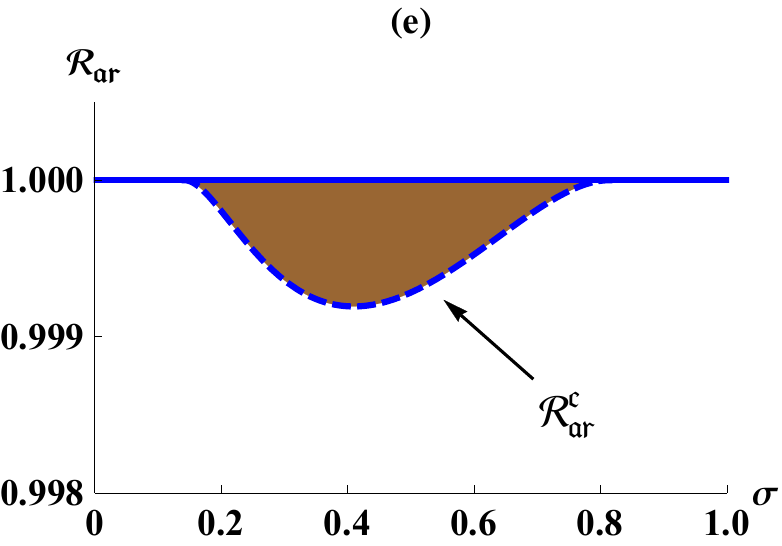}\quad
\includegraphics[width=0.45\textwidth]{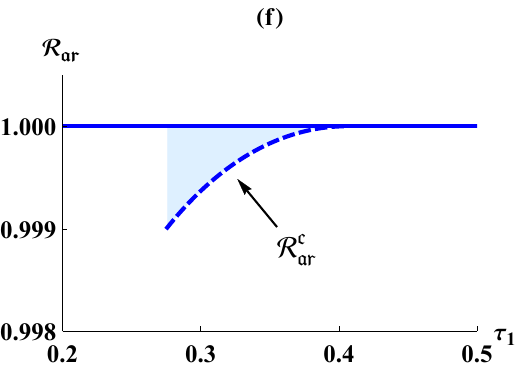}
\caption{Dependence of the lower bound $\mathcal{R}_{ar}^c$ on the parameters $m$, $\mu$, $\tau_2$, $\beta$, $\sigma$ and $\tau_1$.}
\label{dependence}
\end{figure}

Figure \ref{dependence} illustrates the dependence of lower bound $\mathcal{R}_{ar}^c$ given by \eqref{Rarc} on parameters $m$, $\mu$, $\tau_2$, $\beta$, $\sigma$ and $\tau_1$, respectively. The dashed curve corresponds to the lower bound $\mathcal{R}_{ar}^c$ of $\mathcal{R}_{ar}$, such that there are two interior equilibria in the shaded region when $\mathcal{R}_{ar}^c<\mathcal{R}_{ar}<1$. Obviously, $\mathcal{R}_{ar}^c$ increases with $m$, and disappears when $m$ exceeds a threshold value, see Figure \ref{dependence}(a). In Figure \ref{dependence}(d), $\mathcal{R}_{ar}^c$ decreases with $\beta$. It is worth noting that when all other parameters are fixed, $\mathcal{R}_{ar}^c$ may exist if and only if $\beta$ reaches some threshold value ($\beta^*$). The dependence of $\mathcal{R}_{ar}^c$ on parameters $\mu$ (see Figure \ref{dependence}(b)), $\tau_2$ (see Figure \ref{dependence}(c)) and $\tau_1$ (see Figure \ref{dependence}(f)) is similar to that on $m$, however, for the parameter $\tau_1$, $\mathcal{R}_{ar}^c$ may exist only if $\tau_1$ reaches some threshold value. The most special one is the dependence of $\mathcal{R}_{ar}^c$ on parameter $\sigma$, see Figure \ref{dependence}(e), that $\mathcal{R}_{ar}^c$ decreases first with $\sigma$, then increases, until $\sigma$ exceeds some threshold value.  The parameters fixed in this figure have the same values as those in Figure \ref{bistable}.

\begin{figure}[h]
\centering
\includegraphics[width=0.6\textwidth]{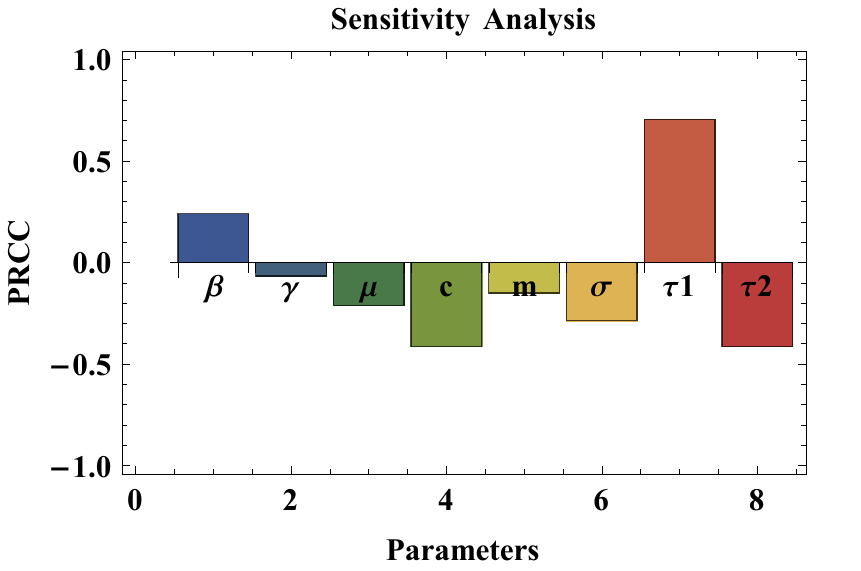}
\caption{Sensitivity analysis of parameters in the model.}
\label{Sensitivity}
\end{figure}

Figure \ref{Sensitivity} exhibits a sensitivity analysis of the model based on Latin Hypercube Sampling. We suppose that $m$ and $\sigma$ obey Triangular distribution[0,1], $\tau_1$ and $\tau_2$ obey Uniform distribution [0,1], $\beta$ obeys Triangular distribution[0,10], $\gamma$ obeys Uniform distribution [1/60,1/30], $\mu$ obeys Uniform distribution [1/20,1/7], and $c$ obeys Triangular distribution[0,0.3]. It is clearly shown that the treatment of drug 1 ($\tau_1$) has the most important influence on the frequency of patients colonized with antibiotic-resistant bacteria, and they are positively correlated, while the treatment of drug 2 is negatively correlated with it.

\section{Discussion}\label{discussion}

Antibiotic resistance has caused worldwide concern, and the studies on the transmission
dynamics of antibiotic-resistant bacteria are important for us to understand and take measures to control it. In the present paper, we give a full analysis of the global dynamics of an antibiotic-resistance model, in which superinfection is included,
and numerical simulations are done to verify the analytical results. Our results suggest that when there is no superinfection ($\sigma=0$), the system has a unique interior equilibrium if and only if $\mathcal{R}_{ar}>1$. That is to say, the model can't generate the backward bifurcation without superinfection. Thus, superinfection ($\sigma>0$) increases the possibility of occurrence of the backward bifurcation.

From model analysis and simulations, we find some threshold conditions for the elimination of transmission of the antibiotic-resistant bacteria, for example, controlling $\mathcal{R}_{rc}\leq1$ such that $\mathcal{R}_{ar}<1$, or controlling $\mathcal{R}_{ar}<\mathcal{R}_{ar}^c$ and so on. In practice, it means control measures can be aimed at reducing the primary transmission rate ($\beta$) or increasing the use of drug 2 ($\tau_{2}$) with cautions that no new resistance develops.

The occurrence of backward bifurcation increases the difficulty to control the resistance transmission. In this situation, a control measure also can be aimed at increasing $\mathcal{R}_{ar}^c$. It is necessary to identify the factors that influence the size of $\mathcal{R}_{ar}^c$. We have shown the dependence of $\mathcal{R}_{ar}^c$ on parameters $m$, $\mu$, $\tau_2$, $\beta$, $\sigma$ and $\tau_1$ in Figure \ref{dependence}. Obviously, the increase of the proportion of admitted already colonized with sensitive bacteria will increase $\mathcal{R}_{ar}^c$, and thus this reduces the size of the window $(\mathcal{R}_{ar}^c,1)$ for multiple interior equilibria; it will have the same effect to increase the treatment rate of drug 2 properly, reduce the primary transmission rate, or reduce the average length of stay in hospital, etc.

In the Appendix, we also provide the global dynamics of the model for $m=0$. Obviously, they have enormous difference with the case $m>0$.

In future study, we will consider the isolation of patients with antibiotic-resistant bacteria, and investigate the dynamics of the new system; or introduce a nonlinear term, for example, assuming that patients with antibiotic-resistant bacteria consume more medical resources, and this leads to a reduction in new patients admitted with sensitive bacteria, that is, the proportion of admitted already colonized with sensitive bacteria $m$ depends on the frequency of individuals who carry the antibiotic-resistant bacteria $R$. More complicated and rich dynamics and bifurcations are expected.

\section*{Appendix}{\label{appendix}}
\noindent{\bfseries{Proof of Proposition \ref{invar_set}}.} Consider the directions of the orbits of system \eqref{plane} on the boundary of  $\Omega$.

Firstly, along the $R$-axis, it follows from the first equation of \eqref{plane} that  $S'(t)|_{S=0}=m\mu\geq 0$. Hence, either $R$-axis is an invariant line of system \eqref{plane} ($m=0$), or the vector field of system \eqref{plane} at each point on $R$-axis points to the interior of $\Omega$ ($m>0$).

Secondly, consider the boundary $S+R=1$.  By system \eqref{plane}, we get that
\begin{equation*}
\frac{d(S+R)}{dt}\Big|_{R=1-S}=-(1-m)\mu-\tau _2-\gamma-\tau_1 S<0.
\end{equation*}
Therefore, any orbit starting from a point on the boundary  $S+R=1$ enters the interior of $\Omega$ as $t$ increases.

Note that $S$-axis is an invariant line of system \eqref{plane}.
By the discussions above, we conclude that $\Omega$ is a positively invariant set.

\noindent{\bfseries{Proof of Theorem \ref{th:E0}.}}
The Jacobian matrix at the resistance-free equilibrium $E_{0}^*$ is:
\[
\left(\begin{array}{cc}
-\sqrt{(\beta-\tau_{1}-\tau_{2}-\gamma-\mu)^{2}+4\beta m\mu} & -\beta(1-c\sigma)S_0\\
0 & (\tau_{2}+\gamma+\mu+\sigma\beta S_{0})(\mathcal{R}_{ar}-1)
\end{array}\right).
\]
Thus, $E_0^*$ has a negative eigenvalue $-\sqrt{(\beta-\tau_{1}-\tau_{2}-\gamma-\mu)^{2}+4\beta m\mu}$, and the other eigenvalue depends on $\mathcal{R}_{ar}$:
\[
(\tau_{2}+\gamma+\mu+\sigma\beta S_{0})(\mathcal{R}_{ar}-1).
\]

Obviously, if $\mathcal{R}_{ar}<1$, then $E_0^*$ is a l.a.s. node, and if $\mathcal{R}_{ar}>1$, then $E_0^*$ is an unstable saddle. When $\mathcal{R}_{ar}=1$, the eigenvalue dependent on $\mathcal{R}_{ar}$ equals to 0, and thus
we need to apply center manifold theorem \cite{ZDHD} to judge the stability of $E_0^*$. First, translate $E_0^*$ to the origin, and do the affine transformation
\[
(S,R)\rightarrow\left(S+\dfrac{\beta (1-\sigma)S_0 R}{\sqrt{(\beta-\tau_{1}-\tau_{2}-\gamma-\mu)^{2}+4\beta m\mu}},R\right).
\]
Then, system \eqref{plane} can be reduced to the norm form:
\begin{equation*}
\begin{split}
\dfrac{dS}{dt} & = -\sqrt{(\beta-\tau_{1}-\tau_{2}-\gamma-\mu)^{2}+4\beta m\mu}S+\mathcal F(S,R),\\
\dfrac{dR}{dt} & =\mathcal G(S,R),
\end{split}
\end{equation*}
where
\begin{equation*}\begin{split}
\mathcal F(S,R)=&-\beta S^2+\dfrac{\beta (1-c\sigma)[\beta-\tau_{1}-\tau_{2}-\gamma-\mu-\beta(1-c+c\sigma)S_0]SR}{\sqrt{(\beta-\tau_{1}-\tau_{2}-\gamma-\mu)^{2}+4\beta m\mu}}\\
&+\dfrac{\beta^2 c(1-\sigma)(1-c\sigma)[\tau_{1}+\tau_{2}+\gamma+\mu-\beta+\beta(1+c\sigma)S_0]S_0 R^2}{(\beta-\tau_{1}-\tau_{2}-\gamma-\mu)^{2}+4\beta m\mu},\\
\end{split}
\end{equation*}
\begin{equation*}\begin{split}
\mathcal G(S,R)=
&-\beta(1-c+c\sigma)SR\\
&+\dfrac{\beta[(1-c)(\beta-\tau_{1}-\tau_{2}-\gamma-\mu)-\beta(1-c-c^2\sigma+c^2\sigma^2)S_0]R^2}{\sqrt{(\beta-\tau_{1}-\tau_{2}-\gamma-\mu)^{2}+4\beta m\mu}}
\end{split}
\end{equation*}
satisfy
\[
\mathcal F(0,0)=0,\quad \mathcal G(0,0)=0,\quad D\mathcal F(0,0)=\mathrm{O},\quad D\mathcal G(0,0)=\mathrm{O}.
\]
It follows from the center manifold theorem that there exists a locally invariant manifold  $S=\mathcal H(R)\in\mathcal{C}^{2}$, such that all the solutions on this manifold have the property:
\begin{equation*}
\dfrac{dR}{dt}=\mathcal G(\mathcal H(R),R)=h_2 R^2+h_3 R^3+O(|R|^4).
\end{equation*}
Here
\begin{equation*}\begin{split}
h_2&=\dfrac{\beta[(1-c)(\beta-\tau_{1}-\tau_{2}-\gamma-\mu)-\beta(1-c-c^2\sigma+c^2\sigma^2)S_0]}{\sqrt{(\beta-\tau_{1}-\tau_{2}-\gamma-\mu)^{2}+4\beta m\mu}},\\
h_3&=\dfrac{\beta^3 c(1-\sigma)(1-c\sigma)(1-c+c\sigma)[\beta-\tau_{1}-\tau_{2}-\gamma-\mu-\beta(1+c\sigma)S_0]S_0}{(\sqrt{(\beta-\tau_{1}-\tau_{2}-\gamma-\mu)^{2}+4\beta m\mu})^3}.
\end{split}\end{equation*}
Using
\[
S_0=\frac{(1-c)(\mathcal{R}_{rc}-1)}{(1-c+c\sigma)\mathcal{R}_{rc}},
\]
which comes from $\mathcal{R}_{ar}=1$, we get
\begin{equation*}\begin{split}\label{h}
h_2&=\dfrac{\beta[\beta(1-c)\sigma(1-\sigma)c^2(\mathcal{R}_{rc}-1)^2-m\mu(1-c+\sigma c)^2\mathcal{R}_{rc}^2]}{(1-c+\sigma c)\mathcal{R}_{rc}(\mathcal{R}_{rc}-1)\sqrt{(\beta-\tau_{1}-\tau_{2}-\gamma-\mu)^{2}+4\beta m\mu}}\\
&=\dfrac{\beta[c\sqrt{\beta(1-c)\sigma(1-\sigma)}(\mathcal{R}_{rc}-1)+(1-c+\sigma c)\sqrt{m\mu}\mathcal{R}_{rc}]f(\mathcal{R}_{rc})}{(1-c+\sigma c)\mathcal{R}_{rc}(\mathcal{R}_{rc}-1)\sqrt{(\beta-\tau_{1}-\tau_{2}-\gamma-\mu)^{2}+4\beta m\mu}},\\
\end{split}
\end{equation*}
where $f(\mathcal{R}_{rc})$ is defined in \eqref{f}. Obviously, $h_2\leq0(>0)$ iff $f(\mathcal{R}_{rc})\leq0(>0)$.

When $h_2\neq0$, i.e., $f(\mathcal{R}_{rc})\neq0$, $E_0^*$ is a saddle-node. However, we are only interested in the directions of the trajectories near $E_0^*$ in the region $\Omega$ (This is because that $E_0^*$ is a boundary equilibrium). Therefore, we get that when $f(\mathcal{R}_{rc})>0$, $E_0^*$ is unstable, when $f(\mathcal{R}_{rc})<0$, $E_0^*$ is l.a.s., and when $f(\mathcal{R}_{rc})=0$, $h_3$ is simplified as
\[
h_3=-\dfrac{\beta^4c^2\sigma(1-\sigma)(1-c+c\sigma)(1-c\sigma)^2S_0^2}{(1-c)(\sqrt{(\beta-\tau_{1}-\tau_{2}-\gamma-\mu)^{2}+4\beta m\mu})^3}<0,
\]
and thus, $E_0^*$ is a l.a.s. node.

\noindent{\bfseries{Proof of Theorem \ref{th:local}.}}
(a) When $\mathcal{R}_{ar}>1$, there are two cases.

If $c\sigma(1-\sigma)=0$, then $\mathrm{det}J(E^{*})=\beta(1-c)m\mu R^{*}/S^{*}>0$. With \eqref{det} and \eqref{delta}, we get that Jacobian matrix $J(E^{*})$ has two negative eigenvalues, i.e., the unique interior equilibrium is a node and it is l.a.s..

If $c\sigma(1-\sigma)>0$, then the unique zero of $\Phi(S)$ in $(0,a^{*})$ is $S^{*}=(-\phi_{1}-\sqrt{\Delta})/(2\phi_{2})$.
Thus
\[
\mathrm{det}J(E^{*})=\frac{\beta R^{*}(\phi_{0}-\phi_{2}S^{*2})}{S^{*}}=-\dfrac{\sqrt{\Delta}(\phi_{1}+\sqrt{\Delta})}{2\phi_{2}}\cdot\frac{\beta R^{*}}{S^{*}}=\beta R^{*}\sqrt{\Delta}>0.
\]
The same conclusion can be obtained using \eqref{det} and \eqref{delta}.

(b)(i) In this case, $\Phi(S)$ has two zeros $S_{1}^{*}=(-\phi_{1}-\sqrt{\Delta})/(2\phi_{2})$ and
$S_{2}^{*}=(-\phi_{1}+\sqrt{\Delta})/(2\phi_{2})$. Denote the corresponding interior equilibria by $E_1^*(S_1^*,R_1^*)$ and  $E_2^*(S_2^*,R_2^*)$, respectively. Thus,
\[
\mathrm{det}J(E_1^{*})=\frac{\beta R_1^{*}(\phi_{0}-\phi_{2}S_1^{*2})}{S_1^{*}}=\beta R_1^{*}\sqrt{\Delta}>0,
\]
and
\[
\mathrm{det}J(E_2^{*})=\frac{\beta R_2^{*}(\phi_{0}-\phi_{2}S_2^{*2})}{S_2^{*}}=-\beta R_2^{*}\sqrt{\Delta}<0.
\]
We have $E_{1}^{*}$ is a l.a.s. node, and $E_{2}^{*}$ is an unstable saddle by \eqref{det} and \eqref{delta}.

(b)(ii) Here $\Phi(S)$ has a zero $S^{*}=-\phi_{1}/(2\phi_{2})$ in $(0,a^{*})$ with multiplicity 2. Further calculus shows that
\[
\mathrm{det}J(E^{*})=\dfrac{\beta R^{*}(4\phi_{0}\phi_{2}-\phi_{1}^{2})}{4\phi_{2}S^{*}}=0.
\]
Therefore, $\mathrm{det}J(E^{*})$ has an eigenvalue equal to zero.
We will determine the stability by using the center manifold theory
to system \eqref{plane}. Translate $E^{*}$ to the origin and take
the affine transformation
\[
(S,R)\rightarrow\left(\dfrac{(1-c+c\sigma)S+(1-c)R}{(1-c\sigma)(1-c+c\sigma)S^{*}+(1-c)^{2}R^{*}},\dfrac{-(1-c)R^{*}S+(1-c\sigma)S^{*}R}{(1-c\sigma)(1-c+c\sigma)S^{*}+(1-c)^{2}R^{*}}\right).
\]
Then system \eqref{plane} is reduced to the norm form:
\begin{equation*}
\begin{split}
\dfrac{dS}{dt} & = (-\beta S^{*}-m\mu/S^{*}-\beta(1-c)R^{*})S+F(S,R),\\
\dfrac{dR}{dt} & = G(S,R),
\end{split}
\end{equation*}
where $F(S,R)$ and $G(S,R)$ are homogenous quadratic polynomials
and thus they satisfy
\[
F(0,0)=0,\quad G(0,0)=0,\quad DF(0,0)=\mathrm{O},\quad DG(0,0)=\mathrm{O}.
\]
It follows from the center manifold theorem \cite{ZDHD} that there
exists a locally invariant manifold $S=h(R)\in\mathcal{C}^{2}$ such
that all the solutions on this center manifold have the property
\begin{equation}
\dfrac{dR}{dt}=-\dfrac{\beta\sigma(1-\sigma)c^{2}(1-c)^{2}R^{*}}{(1-c\sigma)(1-c+c\sigma)S^{*}+(1-c)^{2}R^{*}}R^{2}+O(|E|^{3}).\label{d}
\end{equation}
Since the coefficient of $R^{2}$ in \eqref{d} is negative, we have
$E^{*}$ is a saddle-node and it is unstable.

(b)(iii) In this case $\Delta=(\phi_1+2\phi_2a^*)^2>0$ and $\Phi(S)$ has a unique zero $S^{*}=(-\phi_{1}-\sqrt{\Delta})/(2\phi_{2})$ in $(0,a^{*})$. Hence
\[
\mathrm{det}J(E^{*})=\frac{\beta R^{*}(\phi_{0}-\phi_{2}S^{*2})}{S^{*}}=-\dfrac{\sqrt{\Delta}(\phi_{1}+\sqrt{\Delta})}{2\phi_{2}}\cdot\frac{\beta R^{*}}{S^{*}}=\beta R^{*}\sqrt{\Delta}>0.
\]
In the same way, the unique interior equilibrium is a l.a.s. node.

\noindent{\bfseries{Global dynamics of system \eqref{plane} for the case $m=0$.}}

First, we give the existence theorem of equilibria.
\begin{theorem}\label{Th:S0} Let $\mathcal{R}_{sc}$ and $\mathcal{R}_{rc}$ be the control reproduction numbers defined in \eqref{Rsc} and \eqref{Rrc}, respectively.
\begin{itemize}
    \item[(a)] When $\mathcal{R}_{sc}\leq1$ and $\mathcal{R}_{rc}\leq1$, system \eqref{plane} only has one boundary equilibrium $O(0,0)$, which is the unique disease-free equilibrium;
    \item[(b)] when $\mathcal{R}_{sc}>1$ and $\mathcal{R}_{rc}\leq1$, system \eqref{plane} has a disease-free equilibrium $O(0,0)$ and a resistance-free equilibrium $E_0(S_0,0)$;
    \item[(c)] when $\mathcal{R}_{sc}\leq1$ and $\mathcal{R}_{rc}>1$, system \eqref{plane} has a disease-free equilibrium $O(0,0)$ and a resistant equilibrium $U_0(0,R_0)$;
    \item[(d)] when $\mathcal{R}_{sc}>1$ and $\mathcal{R}_{rc}>1$, system \eqref{plane} always has a disease-free equilibrium $O(0,0)$, a resistance-free equilibrium $E_0(S_0,0)$, and a resistant equilibrium $U_0(0,R_0)$, in addition, \begin{itemize}
     \item[(i)] if $c\sigma(1-\sigma)>0$, $(1-c+c\sigma)\mathcal{R}_{rc}/(1-c+c\sigma\mathcal{R}_{rc})<\mathcal{R}_{sc}<\mathcal{R}_{rc}/(1-c\sigma+c\sigma\mathcal{R}_{rc})$, then system \eqref{plane} has a unique interior equilibrium $E_1(S_1,R_1)$,
     \item[(ii)] if $c\sigma=0$ and $\mathcal{R}_{sc}=\mathcal{R}_{rc}$, then system \eqref{plane} has a singular line $S_0-R-S=0$; and if $\sigma=1$ and $\mathcal{R}_{sc}=\mathcal{R}_{rc}/(1-c+c\mathcal{R}_{rc})$, then system \eqref{plane} has a singular line $S_0-(1-c)R-S=0$,
     \end{itemize}
     where
     \begin{equation}\begin{split}\label{S0S1}
     S_0&=1-\dfrac{1}{\mathcal{R}_{sc}},\quad \,\,\, R_0=1-\dfrac{1}{\mathcal{R}_{rc}},\\
     S_1&=\dfrac{(1-c)(-S_0+(1-c\sigma)R_0)}{c^2\sigma(1-\sigma)},\\
     R_1&=\dfrac{(1-c+c\sigma)S_0-(1-c)R_0}{c^2\sigma(1-\sigma)}.
     \end{split}
     \end{equation}
\end{itemize}
\end{theorem}
Next, we determine the stability of the equilibria described in Theorem \ref{Th:S0}.
\begin{theorem}\label{Th:local}
Consider the equilibria described in Theorem \ref{Th:S0}, the following results hold.
\begin{itemize}
    \item[(a)] When $\mathcal{R}_{sc}\leq1$ and $\mathcal{R}_{rc}\leq1$, the disease-free equilibrium $O(0,0)$ is locally asymptotically stable (l.a.s.);
    \item[(b)] when $\mathcal{R}_{sc}>1$ and $\mathcal{R}_{rc}\leq1$, the disease-free equilibrium $O(0,0)$ is unstable, and the resistance-free equilibrium $E_0(S_0,0)$ is l.a.s.;
    \item[(c)] when $\mathcal{R}_{sc}\leq1$ and $\mathcal{R}_{rc}>1$, the disease-free equilibrium $O(0,0)$ is unstable, and the resistant equilibrium $U_0(0,R_0)$ is l.a.s.;
    \item[(d)] when $\mathcal{R}_{sc}>1$ and $\mathcal{R}_{rc}>1$, the disease-free equilibrium $O(0,0)$ is unstable,
    \begin{itemize}
    \item[(i)] if $c\sigma(1-\sigma)>0$ and $(1-c+c\sigma)\mathcal{R}_{rc}/(1-c+c\sigma\mathcal{R}_{rc})<\mathcal{R}_{sc}<\mathcal{R}_{rc}/(1-c\sigma+c\sigma\mathcal{R}_{rc})$, then both the resistance-free equilibrium $E_0(S_0,0)$ and the resistant equilibrium $U_0(0,R_0)$ are l.a.s., and the  unique interior equilibrium $E_1(S_1,R_1)$ is unstable,
     \item[(ii)] if $c\sigma=0$ and $\mathcal{R}_{sc}=\mathcal{R}_{rc}$, then all the solution orbits of system \eqref{plane} converge to the singular line $S_0-R-S=0$, and if $\sigma=1$ and $\mathcal{R}_{sc}=\mathcal{R}_{rc}/(1-c+c\mathcal{R}_{rc})$, then all the solution orbits of system \eqref{plane} converge to the singular line $S_0-(1-c)R-S=0$,
     \item[(iii)] if $\mathcal{R}_{sc}<(1-c+c\sigma)\mathcal{R}_{rc}/(1-c+c\sigma\mathcal{R}_{rc})$, then the resistance-free equilibrium $E_0(S_0,0)$ is unstable, and the resistant equilibrium $U_0(0,R_0)$ is l.a.s.,
     \item[(iv)] if $\mathcal{R}_{sc}>\mathcal{R}_{rc}/(1-c\sigma+c\sigma\mathcal{R}_{rc})$, then the resistance-free equilibrium $E_0(S_0,0)$ is l.a.s., and the resistant equilibrium $U_0(0,R_0)$ is unstable,
     \item[(v)] if $c\sigma(1-\sigma)>0$ and $\mathcal{R}_{sc}=(1-c+c\sigma)\mathcal{R}_{rc}/(1-c+c\sigma\mathcal{R}_{rc})$, then the resistance-free equilibrium $E_0(S_0,0)$ is unstable and the resistant equilibrium $U_0(0,R_0)$ is l.a.s.,
     \item[(vi)] if $c\sigma(1-\sigma)>0$ and $\mathcal{R}_{sc}=\mathcal{R}_{rc}/(1-c\sigma+c\sigma\mathcal{R}_{rc})$, then the resistance-free equilibrium $E(S_0,0)$ is l.a.s. and the resistant equilibrium $U_0(0,R_0)$ is unstable,
     \end{itemize}
     where $S_0,R_0,S_1$ and $R_1$ are given by \eqref{S0S1}.
     \end{itemize}
\end{theorem}
By Proposition \ref{invar_set}, Theorems \ref{Th:S0}, \ref{Th:local} and Proposition \ref{global}, we obtain the global dynamics of system \eqref{plane} for the case $m=0$, see Figure \ref{Fig.9}. The phase portraits correspond to the cases of Theorem \ref{Th:local}.

\begin{figure}[h]
\centering
\subfigure[Case (a)]{
\label{Fig9.sub.1}
\includegraphics[width=0.31\textwidth]{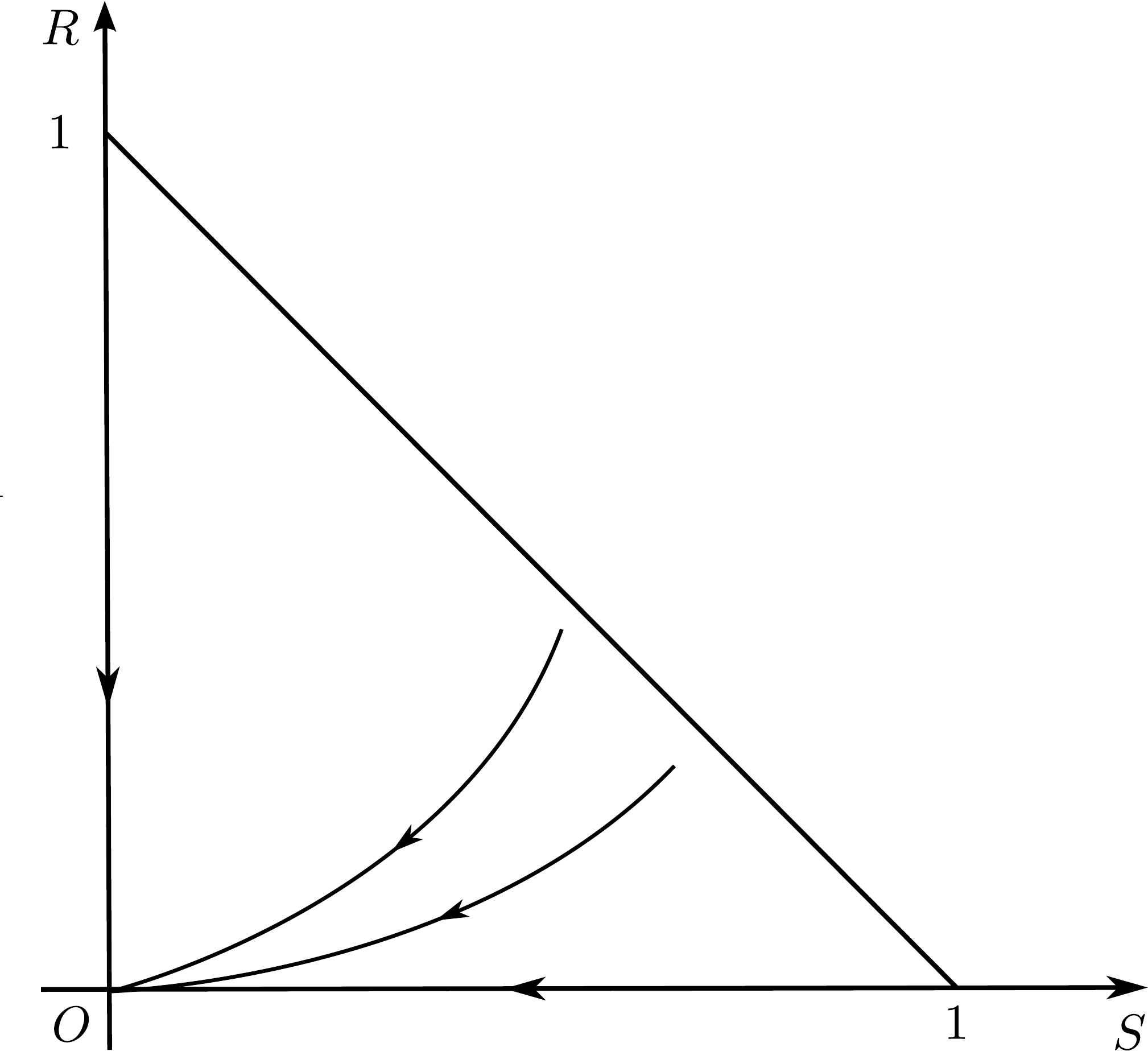}}
\subfigure[Case (b)]{
\label{Fig9.sub.2}
\includegraphics[width=0.31\textwidth]{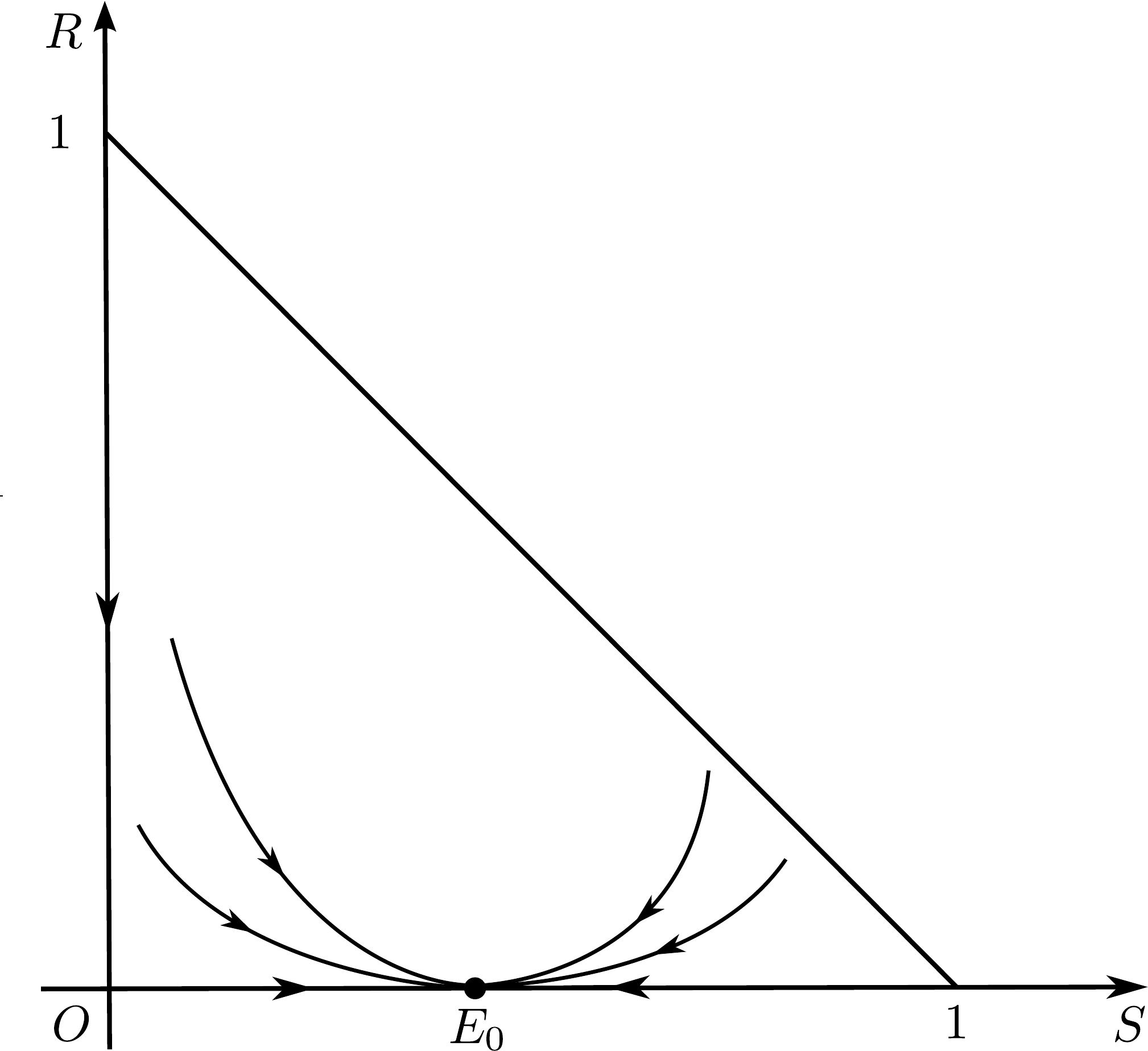}}
\subfigure[Case (c)]{
\label{Fig9.sub.3}
\includegraphics[width=0.31\textwidth]{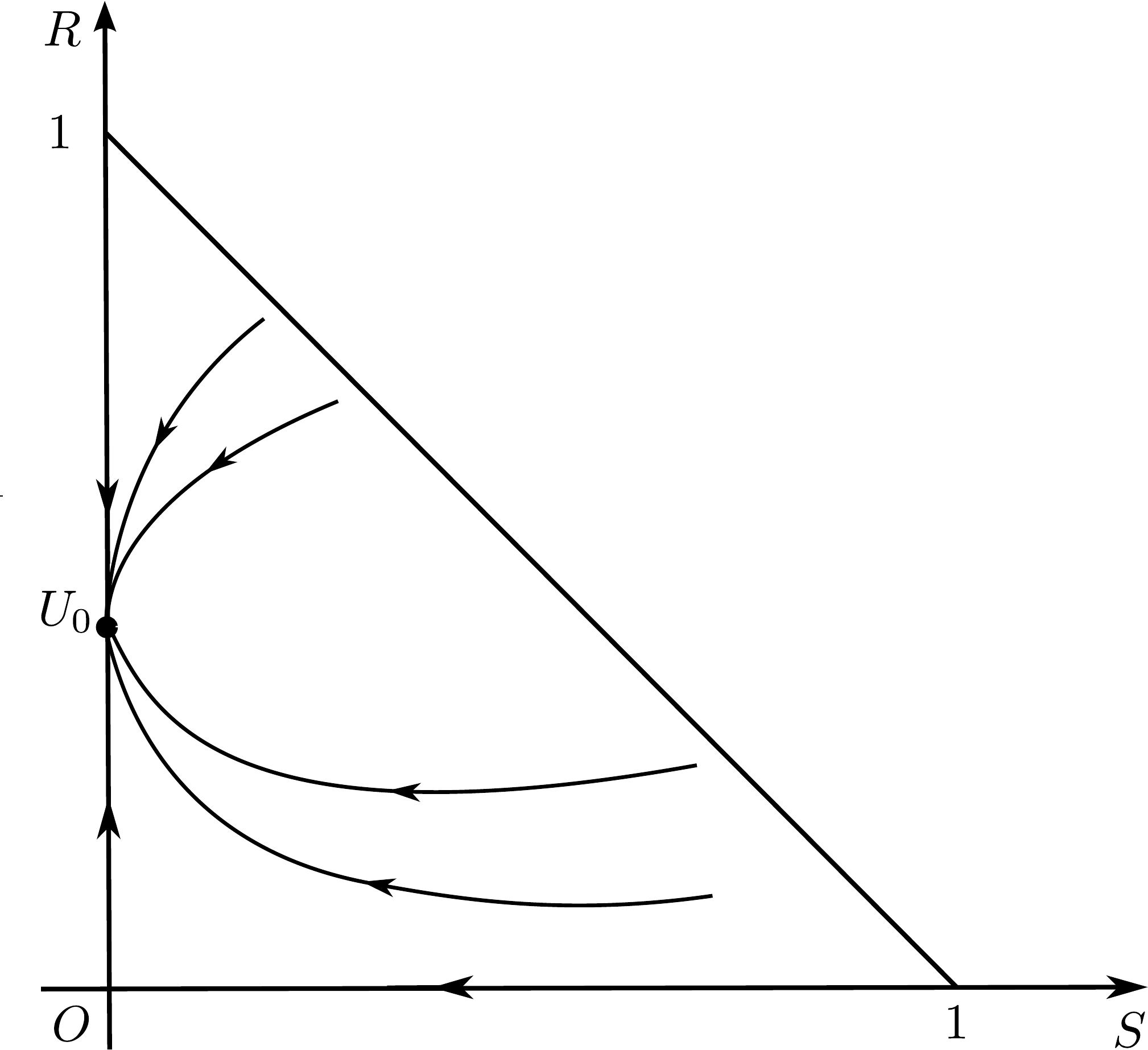}}
\subfigure[Cases d(iv) and d(vi)]{
\label{Fig9.sub.4}
\includegraphics[width=0.31\textwidth]{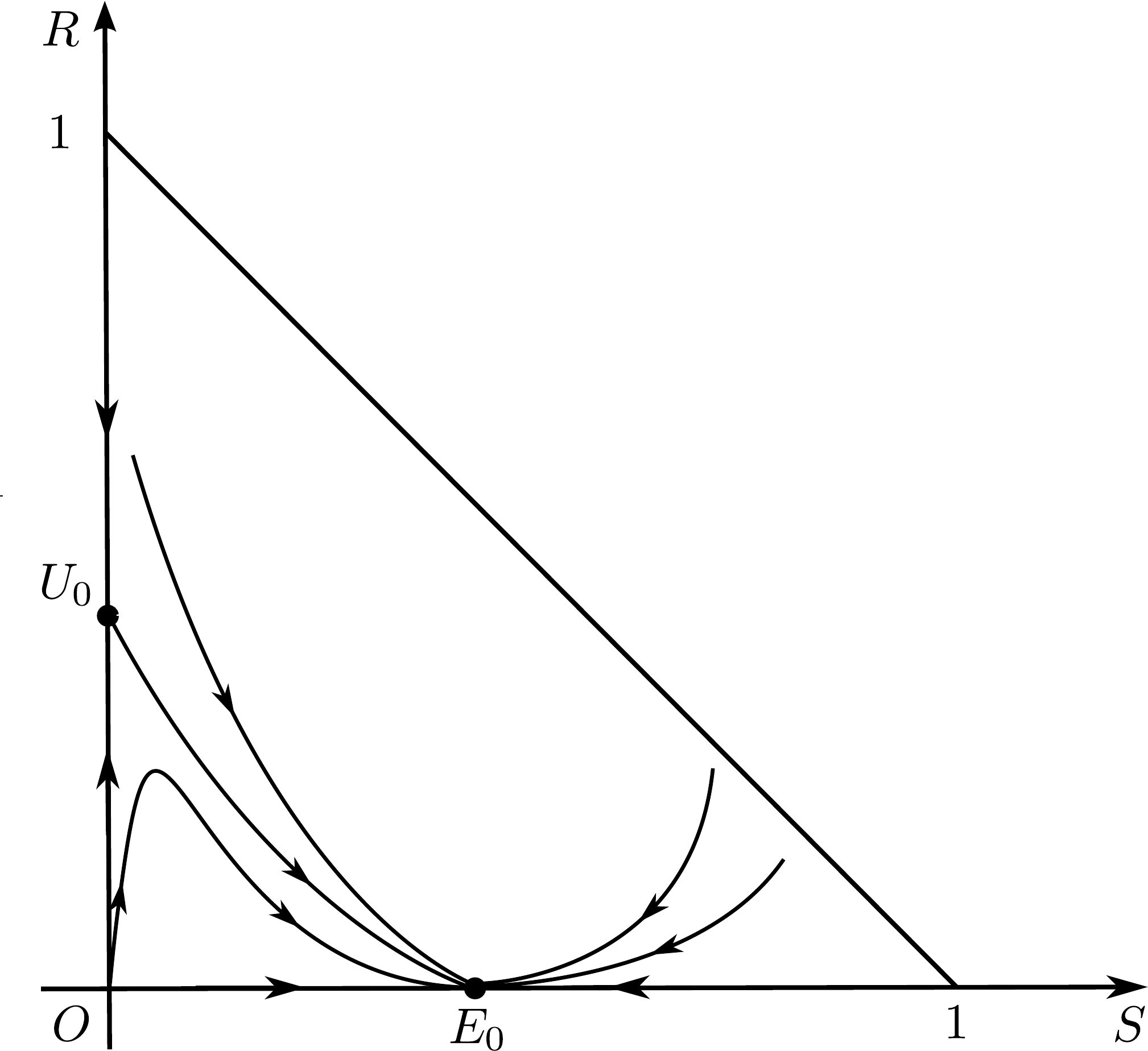}}
\subfigure[Cases d(iii) and d(v)]{
\label{Fig9.sub.5}
\includegraphics[width=0.31\textwidth]{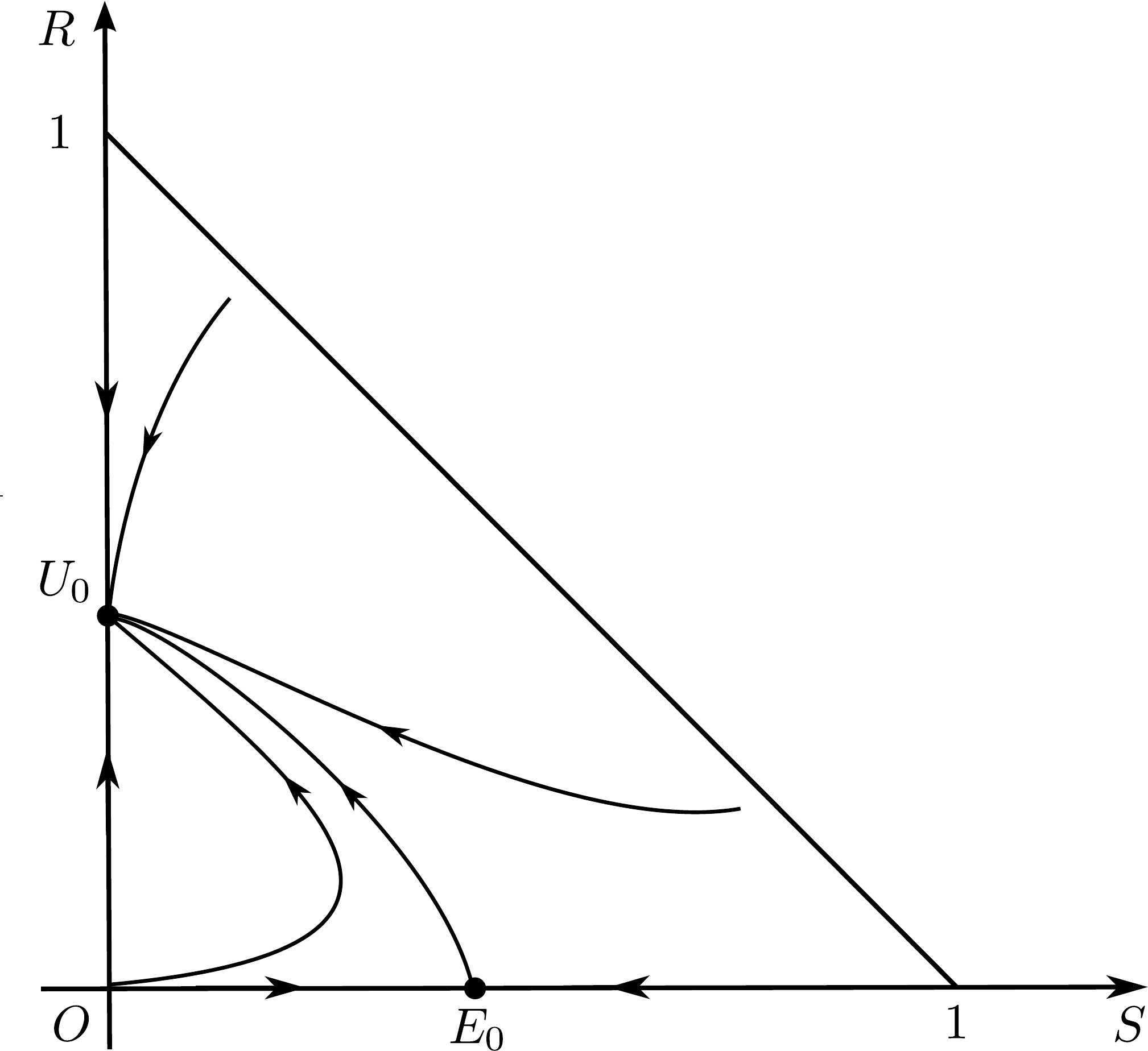}}
\subfigure[Case d(i)]{
\label{Fig9.sub.6}
\includegraphics[width=0.31\textwidth]{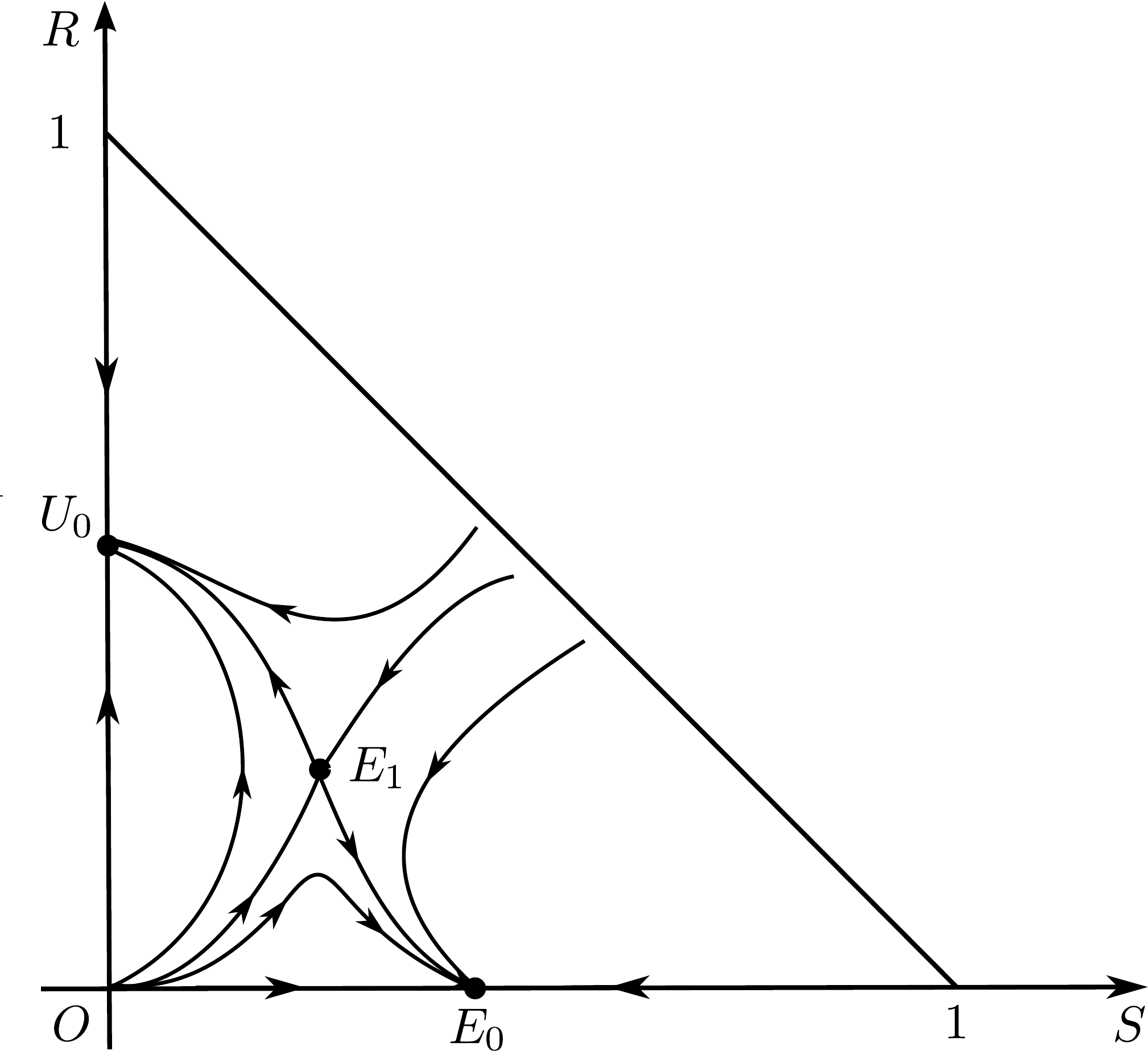}}
\subfigure[Case d(ii)]{
\label{Fig9.sub.7}
\includegraphics[width=0.31\textwidth]{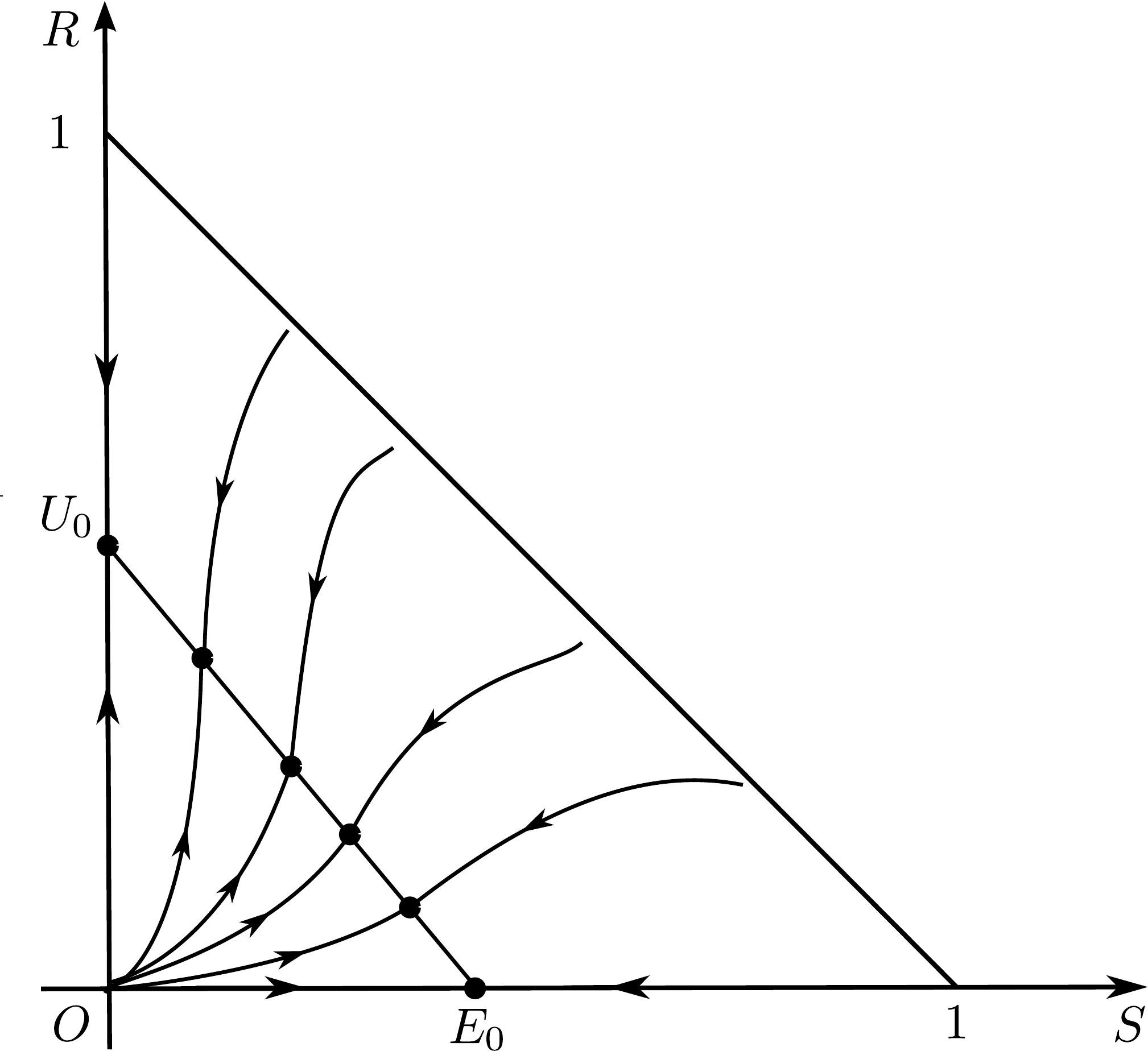}}
\caption{The phase portraits of system \eqref{plane} when $m=0$. }
\label{Fig.9}
\end{figure}
\end{document}